\documentclass[11pt, letterpaper]{article}
\usepackage{fullpage}
\usepackage{amsthm}
\usepackage{amsmath,amssymb,amsfonts,nicefrac}
\usepackage{xspace}
\usepackage{color}
\usepackage{url}
\usepackage{hyperref}
\usepackage{bm}
\usepackage{bbm}
\usepackage{times}

\usepackage{enumitem}

\newtheorem{thm}{Theorem}[section]

\newtheorem{lemma}[thm]{Lemma}
\newtheorem{corollary}[thm]{Corollary}
\newtheorem{claim}[thm]{Claim}
\newtheorem{proposition}[thm]{Proposition}
\newtheorem{definition}[thm]{Definition}
\newtheorem{remark}[thm]{Remark}

\newcommand\E{\mathop{\mathbb{E}}}
\newcommand\card[1]{\left| {#1} \right|}
\newcommand\sett[2]{\left\{ \left. #1 \;\right\vert #2 \right\}}

\newcommand\set[1]{{\left\{ #1 \right\}}}
\newcommand\Prob[2]{{\Pr_{#1}\left[ {#2} \right]}}

\newcommand\cProb[3]{{\Pr_{#1}\left[ \left. #3 \;\right\vert #2 \right]}}

\newcommand\Expect[2]{{\mathop{\mathbb{E}}_{#1}\left[ {#2} \right]}}

\newcommand\norm[1]{\| #1 \|}

\newcommand\defeq{\stackrel{def}{=}}
\newcommand\skipi{{\vskip 10pt}}

\newcommand\inner[2]{\langle{#1},{#2}\rangle}
\newcommand\eps{\varepsilon}

\renewcommand\geq{\geqslant}
\renewcommand\leq{\leqslant}

\newcommand{\M}{\mathcal{M}}

\newcommand{\rom}[1]{\uppercase\expandafter{\romannumeral #1\relax}}

\title{Improved Optimal Testing Results from Global Hypercontractivity}
\author{Tali Kaufman
\thanks{Department of Computer Science, Bar-Ilan University.}
\and Dor Minzer \thanks{Department of Mathematics, Massachusetts Institute of Technology, Cambridge, USA. Supported by a Sloan Research
Fellowship.}}
\date{\vspace{-5ex}}
\begin{document}
\maketitle
\begin{abstract}
The problem of testing low-degree polynomials has received significant attention over the years due to its importance in theoretical computer science,
and in particular in complexity theory. The problem is specified by three parameters: field size $q$, degree $d$ and proximity parameter $\delta$,
and the goal is to design a tester making as few as possible queries to a given function, which is able to distinguish between the case the given function has
degree at most $d$, and the case the given function is $\delta$-far from any degree $d$ function.

With respect to these parameters, we say that a tester is {\it optimal} if it makes $O(q^d+1/\delta)$ queries (which are known to be necessary).
For the field of size $q$, such tester was first given by Bhattacharyya et al.~for $q=2$, and later by Haramaty et al.~\cite{HSS} for all prime powers $q$.
In fact, they showed that the natural $t$-flat tester is an optimal tester for the Reed-Muller code, for an appropriate $t$.
%(which is $\lceil\frac{d+1}{q-q/p}\rceil$ where $p$ is prime and $q$ is a power of $p$).
Here, the $t$-flat tester is the tester that picks a uniformly random affine subspace $A$ of dimension $t$, and checks that
${\sf deg}(f|_{A})\leq d$. Their analysis proves that the dependency of the $t$-flat tester on $\delta$ and $d$ is optimal,
however the dependency on the field size, i.e. the hidden constant in the $O$, is a tower-type function in $q$.

We improve the result of Haramaty et al., showing that the dependency on the field size is polynomial. Our technique also
applies in the more general setting of lifted affine invariant codes, and gives the same polynomial dependency on the field size. This
answers a problem raised in~\cite{HRZS}.

Our approach significantly deviates from the strategy taken in earlier works~\cite{BKSSZ,HSS,HRZS}, and is based on studying the structure of the collection of erroneous subspaces,
i.e. subspaces $A$ such that $f|_{A}$ has degree greater than $d$. Towards this end, we observe that these sets are poorly expanding in the affine version of
the Grassmann graph and use that to establish structural results on them via global hypercontractivity. We then use this structure to perform local correction on $f$.
\end{abstract}
\section{Introduction}
The Reed-Muller code is one of the most basic and useful codes in theoretical computer science. A key aspect of the Reed-Muller code, which plays
a significant role in its applications to complexity theory and in particular in the construction of probabilistically checkable proofs, is the
its local testability. Namely, given a truth table of a function over a field, we wish to be able to distinguish between the case that this truth table
represents a Reed-Muller codeword, i.e. a low degree function, and the case it is far from any Reed-Muller codeword.\footnote{Variations of this problems exists,
such as when instead of giving the truth table of a function, one is given a table of supposed restrictions of the function to higher dimensional objects
such as lines or planes; see for example~\cite{RazSafra}.}

Usually, the notion of local-testability of the Reed-Muller codes asserts that when the degree $d$, the field size $q$ and proximity parameter $\delta$ are all
thought of as constants, then there is a tester whose query complexity is constant. With regards to this definition, earlier works~\cite{AKKLR,KR,JPRZ} showed that the
Reed-Muller code is testable. The current work is mainly concerned with the stronger notion of {\it optimal testers} for the Reed-Muller codes.
Here, we wish to get a tester whose query complexity is tight with respect to $d$, $q$ and $\delta$ when they are not thought of as constant.
A typical setting to think about is when the proximity parameter is fairly small, $0<\delta\leq q^{-d}$, in which case it is clear that any tester for the corresponding Reed-Muller code must make at least $\Omega(1/\delta)$
queries.

With respect to this notion, it was shown that the Reed-Muller codes are optimally testable: first in~\cite{BKSSZ} for $\mathbb{F}_2$,
and then to general fields~\cite{HSS}. These works get an optimal dependency of the query complexity on the degree parameter
and the proximity parameter, however they only apply in the case the field size is relatively small. Indeed, the dependency of the rejection probability on
the field size is inverse tower-type, which stems from the fact that their proof utilizes the Density Hales-Jewett theorem.

The main result of this paper is an improvement of the above mentioned results, getting an optimal dependency on the degree parameter while simultaneously
getting a polynomial dependency of the field size. Our approach significantly deviates from the strategy taken in~\cite{BKSSZ,HSS},
and is based on studying expansion properties in the associated affine Grassmann graph. As a side contribution, we prove versions of expansion theorems
that were previously shown for the Grassmann graph~\cite{DKKMS2,KMS2} to the affine Grassmann graph, which turns out to include slight additional complications.

We hope the approach presented herein could be useful in proving optimal testing results for other codes. Indeed,
while our argument does use specific properties of polynomials, it does so ``minimally" and the structure of the underlying
structure of the queries plays a more important role.

\subsection{Local testability of Reed-Muller codes}
Throughout this paper, $p$ denotes a prime number and $q$ denotes a power of $p$.
\begin{definition}
  For a function $f\colon \mathbb{F}_q^n\to\mathbb{F}_q$, we denote
  \[
  \delta_d(f) = \min_{\substack{g\colon\mathbb{F}_q^n\to\mathbb{F}_q \\ \text{ of degree $d$}}} \Prob{x\in \mathbb{F}_q^n}{f(x)\neq g(x)}.
  \]
\end{definition}

In this paper, we consider the $t$-flat tester which is parameterized by a dimension $t$. Here and throughout, a $t$-flat of a given vector
space $W$ (say $W = \mathbb{F}_q^n$) is a $t$-dimensional affine subspace of it.
The $t$-flat tester works by sampling a random $t$-flat $T\subseteq \mathbb{F}_q^n$, and checking that $f|_{T}$ has degree at most $d$.
The $t$ that we pick is the minimal one that makes sense -- i.e. the minimal $t$ such that each $f\colon\mathbb{F}_q^n\to\mathbb{F}_q$
of degree larger than $d$ fails the test with positive probability, which turns out to be $t = \lceil\frac{d+1}{q-q/p}\rceil$~\cite{JPRZ,KR}.

\begin{definition}
  Given a function $f\colon \mathbb{F}_q^n\to\mathbb{F}_q$, and $t,d\in\mathbb{N}$,
  the $t$-flat test proceeds by picking an affine subspace $T$ of dimension $t$, and testing if $f|_{T}$
  is a degree $d$ polynomial. The rejection probability of this test is denoted by $\eps_{t,d}(f)$.
\end{definition}

Let us focus, for a moment, on the case that $q=2$. In this case the $t$-flat test was first analyzed in~\cite{AKKLR}, who proved that $\eps_{t,d}(f)\geq q^{-t} \delta_d(f)$.
An improved analysis of the tester was given in~\cite{BKSSZ}, who showed that the $t$-flat tester is in fact an optimal tester, and in particular that
$\eps_{t,d}(f)\geq \min(c, q^{t}\delta_d(f))$ for some absolute constant $c>0$.
The result was later generalized to general fields in~\cite{HSS}, which reads:
\begin{thm}\label{thm:HSS}
  For all primes $p$ and $q$ powers of $p$, there is $c(q)>0$ such that for all $d\in\mathbb{N}$ and $f\colon\mathbb{F}_q^n\to\mathbb{F}_q$ it holds that
  for $t=\lceil\frac{d+1}{q-q/p}\rceil$ we have
  \[
  \eps_{t,d}(f)\geq c(q)\min(1,q^{t}\delta_{d}(f)).
  \]
\end{thm}
The analysis of both~\cite{BKSSZ} and~\cite{HSS} follows the same high-level inductive approach on the dimension $n$.
Assuming the rejection probability of a given function $f\colon \mathbb{F}_q^n\to\mathbb{F}_q$ is small,
one considers the restriction of $f$ to hyperplanes (i.e., to subspaces of $\mathbb{F}_q^n$ of co-dimension $1$),
on which the inductive hypothesis gives, for each hyperplane $W$, a candidate degree $d$ function that is close to $f|_W$.
The main task then (both in~\cite{BKSSZ} and in~\cite{HSS}) is to ``sew'' together these candidate functions.
Towards this end, a careful choice of the collection of hyperplanes that are most convenient for the task must be made.
This choice is rather simple for $\mathbb{F}_2$, but becomes much more complex in $\mathbb{F}_q$, and to do so
the authors use Ramsey-type results, more specifically the Density Hales-Jewett theorem.
This ultimately leads to an inverse tower-type bound dependency of $c(q)$ on $q$.

Our main result is an improved quantitative version of Theorem~\ref{thm:HSS}. Namely, we prove:
\begin{thm}\label{thm:main}
  For all primes $p$ a prime power $q$ of $p$, there is $c(q) = q^{-O(1)}$ such that for all $d\in\mathbb{N}$ and $f\colon\mathbb{F}_q^n\to\mathbb{F}_q$ it holds that
  for $t=\lceil\frac{d+1}{q-q/p}\rceil$ we have
  \[
  \eps_{t,d}(f)\geq c(q)\min(1,q^{t}\delta_{d}(f)).
  \]
\end{thm}

Our approach is significantly different from the previously mentioned inductive hypothesis.
At a high level, we consider the set of erroneous $t$-flats, i.e.
$S = \sett{T}{{\sf dim}(T) = t, f|_{T}\text{ is not degree $d$}}$ and establish a structural result on it. Interestingly, our starting point is a lemma from~\cite{HSS}
(which is a variant of a lemma already appearing in~\cite{BKSSZ}), which in our language upper bounds the measure of the upper shadow of $S$ as a function of the measure
of $S$. Here, the upper shadow of $S$ is
\[
S\uparrow = \sett{B}{{\sf dim}(B) = t+1, \exists T\in S\text{ such that }T\subseteq B},
\]
and the lemma from~\cite{HSS} asserts that $\mu(S\uparrow)\leq q\cdot\mu(S)$. In~\cite{BKSSZ,HSS} this lemma is used to relate
the rejection probability of the $t$-flat tester and the $(t+k)$-flat tester; in particular it implies that the
rejection probability of the $t$-flat tester is at least $q^{-k}$ times the rejection probability of the $(t+k)$-flat tester.

We use this lemma in a different way. The point here is that as $S$ is a small set, the condition that $\mu(S\uparrow)\leq q\cdot\mu(S)$ is
already itself very restrictive. Examples of $S$ that exhibit such behaviours can be thought of as the subspace analog of
collections of subsets that are nearly tight for the classical Kruskal-Katona theorem. Indeed, a natural type of such small set is
\[
H_x = \sett{T}{{\sf dim}(T) = t, T\ni x}.
\]
We show (simplifying matters somewhat) that indeed, any small $S$ such that $\mu(S\uparrow)\leq q\mu(S)$ must almost contain a copy of $H_x$ for some $x$.
This suggests that an error occurs at $x$ and
that we should change the value of $f(x)$. Indeed, this is the high level strategy we pursue, and we defer
a more detailed description to Section~\ref{sec:techniques}.

\subsection{Lifted affine invariant linear codes}
Our argument in the proof of Theorem~\ref{thm:main} also applies in the more general setting of lifted affine invariant codes. In this case, an analogous result
to Theorem~\ref{thm:HSS} was proved in~\cite{HRZS} with the same type of inverse-type tower dependency on the field size.
Our proof gives a polynomial dependency on the field size making progress along an open problem raised in~\cite{HRZS}.

To present our result for lifted affine invariant codes, we quickly recall the setting. Let $q$ be a prime power, $t\in\mathbb{N}$ and suppose
$\mathcal{B} \subseteq\set{g\colon\mathbb{F}_q^t\to\mathbb{F}_q}$ is
an affine invariant set of function. By that, we mean that $\mathcal{B}$ is closed under composition with affine transformations. Given
$n\geq t$, the $n$-lift of $\mathcal{B}$ denoted by $\mathcal{F} = {\sf Lift}_n(\mathcal{B})$ is defined as
\[
\mathcal{F} = \sett{f\colon\mathbb{F}_q^n\to\mathbb{F}_q}{\forall\text{ $t$-flats } A\subseteq \mathbb{F}_q^n, f|_{A}\in\mathcal{B}}.
\]
For $k\geq t$, the $k$-flat tester proceeds by taking a $k$-flat $A\subseteq\mathbb{F}_q^n$ randomly, and checking that $f|_{A}\in {\sf Lift}_k(\mathcal{B})$, in
which case we say the test accepts. We denote by $\eps_k(f)$ the probability that the $k$-flat tester rejects on $f$.
\begin{thm}\label{thm:main2}
  For all prime powers $q$, there is $c(q) = q^{-O(1)}$ such that for $t\in\mathbb{N}$,
  if $\mathcal{B}\subseteq\set{g\colon\mathbb{F}_q^t\to\mathbb{F}_q}$
  is an affine invariant linear code, and $f\colon\mathbb{F}_q^n\to\mathbb{F}_q$, then
  \[
  \eps_{t}(f)\geq c(q)\min(1,q^{t}\Delta(f,\mathcal{F})),
  \]
  where $\Delta(f,\mathcal{F})$ is the relative Hamming distance between $f$ and $\mathcal{F}$.
\end{thm}

Using a reduction from~\cite{HRZS}, one may use Theorem~\ref{thm:main2} in order to get the following slightly more general result.
The proof is exactly as in~\cite[Section 7]{HRZS}, and is hence omitted.
\begin{thm}\label{thm:main3}
  For all primes $p$, a power $q$ of $p$, and $Q$ a power of $q$, there is $c(Q) = Q^{-O(1)}$
  such that for $t\in\mathbb{N}$, if $\mathcal{B}\subseteq\set{g\colon\mathbb{F}_Q^t\to\mathbb{F}_q}$
  is an affine invariant linear code, and $f\colon\mathbb{F}_Q^n\to\mathbb{F}_q$, then
  \[
  \eps_{t}(f)\geq c(Q)\min(1,Q^{t}\Delta(f,\mathcal{F})),
  \]
  where $\Delta(f,\mathcal{F})$ is the relative Hamming distance between $f$ and $\mathcal{F}$.
\end{thm}
\subsection{The affine Grassmann graph}
To execute our approach we consider the affine Grassmann graph along with an appropriate random walk on it.
Given an affine space $W$ of dimension $k$ over $\mathbb{F}_q$ and an integer $\ell<k$, the affine Grassmann graph ${\sf AffGras}(W,\ell)$ contains as vertices all $\ell$-flats
of $W$, which we denote by $V(k,\ell)$ when the space $W$ is clear from the context.
The edges of the graph are thought of as weighted according to the following randomized process: starting from a $\ell$-flat $A$, we
take a random $(\ell+1)$-flat satisfying $A\subseteq B\subseteq W$, and then take a random $\ell$-flat $A_2\subseteq B$; the weight of the
edge $(A,A_2)$ is the probability it is sampled by this process.

Our first observation is that combining the lemma from~\cite{HSS} with sharp-threshold type result from~\cite{KLLM}, one concludes that
\[
1-\Phi(S)\geq \frac{1}{q}.
\]
Here, $\Phi(S)$ is the expansion of the set $S$, defined as $\Prob{A\in S, A'\sim A\text{ neighbour}}{A'\not\in S}$. Thus, we would be able to gain significant
insight into the structure of $S$ provided we could give sufficiently good characterization of sets in the affine Grassmann graph that are poorly expanding.
This is exactly the type of question that was studied recently in the context of the $2$-to-$2$ Games Theorem~\cite{KMS1,DKKMS1,DKKMS2,KMS2}, and we leverage
insights gained from there in our case of interest.

The affine variant of the Grassmann graph includes further complications, which we explain next. Roughly speaking, the eigenvalues of it are $q^{-i}$ for $i=0,\ldots,\ell$,
hence it can be shown that for sets $S$ of size smaller than $\eps$, one always has $\Phi(S)\geq 1-1/q - O(\eps)$. Thus, in our case we are interested in studying the
structure of sets $S$ that nearly attain this minimum.

The two very natural analogs of small poorly expanding sets are the analogs of zoom-in and zoom-out sets from the non-affine version of the Grassmann graph,
and are defined as follows. For a vector
$z\in W$ and an affine hyperplane $W'\subseteq W$, and zoom-in with respect to $z$ and the zoom-out with respect to $W'$ sets are defined as
\[
H_z = \sett{A\in V(W,\ell)}{z\in A},
\qquad\qquad
H_{W'} = \sett{A\in V(W,\ell)}{A\subseteq W'}.
\]
It can be shown without much difficulty that $H_z$ and $H_{W'}$ have small fractional size, and that
$1-\Phi(H_z)\geq \frac{1}{q}$, $1-\Phi(H_{W'})\geq \frac{1}{q}$. These are the natural analogs of sets
that were shown in~\cite{KMS2} to capture, in some sense, the structure of all small non-expanding sets
in the Grassmann graph. However, in the affine version of the Grassmann graph there are more examples.

Given $z\in W\setminus\set{0}$ and a hyperplane $W'\subseteq W$, one may consider {\it the zoom in and zoom-out with respect to
the linear part}. To define these, first let us note that given an affine subspace $A\in V(W,\ell)$, one may write
$A = x+A'$ where $A'\subseteq W$ is a linear space, and $x\in A$ is some vector (we note that $A'$ is
unique but $x$ is not). Thus, we may define
\begin{align*}
&H_{z,{\sf lin}} = \sett{A\in V(W,\ell)}{A = x+\tilde{A}\text{ for some $x\in W$ and a linear space $\tilde{A}$},\text{ and } z\in \tilde{A}},\\
&H_{W',{\sf lin}} = \sett{A\in V(W,\ell)}{A = x+\tilde{A}\text{ for some $x\in W$ and a linear space $\tilde{A}$},\text{ and } \tilde{A}\subseteq W'}.
\end{align*}
It is easy to see that these sets are also small and have expansion roughly $1-\frac{1}{q}$, and furthermore that they are ``genuinely'' new examples
(i.e., they are linear combinations of the basic zoom-in/zoom-out sets). We show that in a sense, these examples entirely capture the structure
of sets $S$ with $1-\Phi(S)\geq \frac{1}{q}$.

To be more precise, we say a set $S$ is $\xi$-pseudo-random with respect to zoom-ins/ zoom-outs/ zoom-ins on
the linear part/ zoom outs of the linear part -- say zoom-ins for concreteness -- if $\mu(S\cap H_z)\leq \xi\mu(H_z)$ for all $z\in W$
(see Definition~\ref{def:pseudo_random} for a more formal definition). In this language, our main expansion result, Theorem~\ref{thm:expansion},
asserts that if $S$ is $\xi$-pseudo-random with respect to zoom-outs (standard and on the linear part), and with respect to zoom-ins on the linear
part, and $1-\Phi(S)\geq \frac{1}{q}$, then $S$ is highly non-pseudo-random with respect to zoom-ins. Namely, there is $z$ such that
$\mu(S\cap H_z)\geq (1-o(1))\mu(H_z)$, or in words $S$ almost contains $H_z$; see Theorem~\ref{thm:expansion} for a precise statement.
Here and throughout, $\mu$ represents the uniform measure over $V(W,\ell)$.

With our testing question in mind, this sort of structure appears natural as it suggests that we may want to change the value of $f$  on $z$.
\begin{remark}
A few remarks are in order:
\begin{enumerate}
  \item Our expansion result here is tailored for our application, however our technique can be used to establish weaker structural
  for a small set $S$ so long as $1-\Phi(S)\geq \frac{1}{q^2} + \delta$.

  \item The diligent reader may notice that in the statement above, there is an asymmetric role to each one of the zoom-sets. This is a by-product again of the application
  we have in mind as we can show that for the set $S$ of erroneous subspaces, these pseudo-randomness conditions hold. In more generality though, it may be proved that if
  a set $S$ is very pseudo-random with respect to $3$ of the zoom notions (i.e. $\xi$-pseudo-random where $\xi$ is small), then it is very not pseudo-random with respect
  to the last notion of zoom (i.e. almost containing a copy of such set).

  \item It would be interesting to prove expansion theorems in the affine Grassmann graph in greater generality similarly to the way it was done in~\cite{KMS2}. Namely,
  proving that if $1-\Phi(S)\geq \frac{1}{q^{r}} + \delta$, then $S$ cannot be $\eps = \eps(\delta,r)>0$ pseudo-random with respect to $r$-wise intersections of zoom-sets.
  That is, there must be copies of zoom-sets $H^1,\ldots,H^r$ that intersect non-trivially such that $\mu(S\cap \bigcap_{i=1}^{r} H_i)\geq \xi\mu(\bigcap_{i=1}^r H_i)$.
\end{enumerate}
\end{remark}

\subsection{Our techniques}\label{sec:techniques}
Our expansion theorem is proved using Fourier analysis similarly to~\cite{KMS2} and is deferred to the appendix.
Next, we explain how it is used in order to prove Theorem~\ref{thm:main}.

The proof has two components. We consider the collection of erroneous subspaces, i.e.
\[
S = \sett{A\in {\sf AffGras}(\mathbb{F}_q^n, t)}{f|_{A} \text{ is not of degree $d$}}.
\]
%First, we observe using Lemma~\ref{lem:relate} that the increase of measure of $S$ when moving from $S$ to $S\uparrow$ is very slow, resembling the
%behaviour of families of the form $H_x = \sett{A}{x\in A}$ for $x\in\mathbb{F}_q^n$. We then use this information together with Lemma~\ref{lem:sharp_thresh}
%and Theorem~\ref{thm:expansion} to show that, given $S$ is small enough, it nearly contains a copy of $H_x$ for some vector $x$.
First, as explained earlier we observe that $\mu(S\uparrow)\leq q\mu(S)$, and deduce that
$1-\Phi(S)\geq \frac{1}{q}$. We then wish to apply our expansion theorem, and towards this end we show that  $S$ is
pseudo-random with respect to zoom-outs (both standard and with respect to the linear part), as well as on zoom-ins with respect to the linear part.
Proving pseudo-randomness with respect to zoom-outs is fairly easy as these sets enlarge considerably when taking an upper shadow.
The proof that $S$ is pseudo-random with respect to zoom-ins on the linear part is more tricky. In a sense, the idea is that given a $(t+1)$-flat
$A = x+\tilde{A}$ such that $z\in\tilde{A}$ but otherwise $A$ is random, we may find $t$-flats $B_1,\ldots,B_{t+1}\subseteq A$ such that marginally
each one of them is distributed uniformly, and together they cover $A$ entirely. Thus, as errors do not really accumulate on these $B$'s, they cannot be concentrated on
$A$'s of this type. The formal proof proceeds a bit differently and makes use again of our expansion theorem in a lower-order affine Grassmann graph; see Claim~\ref{claim:pseduo_zoom_in_linear} for details.

We then deduce, using our expansion theorem, that $S$ nearly contains a copy of $H_x$ for some $x\in\mathbb{F}_q^n$.
In words, the test almost always fails if it is being conducted on a subspace $A$ that contains the point $x$.
This suggests that $x$ is a point in which we should change the value of $f$ in order to get closer to a degree $d$ polynomial, and we indeed argue this way.

This is the correction step of the argument. The simplest case is $q=2$, which is instructive to consider.
Indeed, in this case we have that $t=d+1$, and we argue that if we flip the value of the point $x$,
the rejection probability of the tests drops additively by $\Theta(2^{d-n})$. The point here is that if $g$ is a polynomial of degree $d+1$ on
a subspace of dimension $d+1$ over $\mathbb{F}_2$, then flipping any single value of $g$ results in a polynomial of degree at most $d$. Iterating this
argument shows that after we change the values of $f$ on at most $O(2^{n-d} \eps)$ points, the rejection probability drops to $0$, at which point our function must be
a degree $d$ polynomial.

In the more general case of $\mathbb{F}_q$, the correction step is not as simple and requires more work. Here, given such point $x$, we consider a random
affine subspace $A$ of dimension $t+100$ containing $x$. We now focus on affine subspaces $B\subseteq A$ of dimension $t$, and note that expectedly over the
random choice of $A$:
\begin{enumerate}
  \item the fraction of such $B$'s containing $x$ on which $f|_{B}$ is degree $d$ is $O(\eps)$;
  \item the fraction of such $B$'s not containing $x$, on which $f|_B$ is degree $d$, is $1-O(\eps)$.
\end{enumerate}
By Markov's inequality, we have that with probability at least $0.99$ both of these events hold simultaneously. We fix such $A$,
and next claim that provided $\eps$ is sufficiently small (depending only on $q$), we can change the value of $f(x)$ in some way
so that the fraction of $B$'s containing $x$ as in the first item above, would be at least $1/(2q)$ (thus lowering the rejection
probability of the test on such subspaces). We establish that via two steps:

\paragraph{Bootstrapping errors on $B\not\ni x$.}
We show that provided that $\eps$ is small, having chosen $A$ as above, if $f|_{B}$ is degree $d$ for at least
$1-O(\eps)$ fraction of the $t$-flats $B\subseteq A$ not containing $x$, then the test must pass in fact on all of these $t$-flats. Intuitively, the idea here is to consider the random walk
  on $B$'s that moves from $B$ of dimension $t$ to $B'$ of dimension $t+1$ that doesn't contain $x$, and then back to $B''\subseteq B'$ of dimension $t$ and show that,
  as before, due to expansion considerations, the errors must be very structured as zoom ins. However, as $\eps$ is very small, zoom-ins are too large
  and hence the set of errors must be empty.

  The precise execution of this step is done differently, as we do not really wish to study this random walk operation as described above; instead, we look
  at intermediate $(t+50)$-flats $C\subseteq A$ that do not contain $x$, and perform the standard random walk on ${\sf AffGras}(C,t)$.

\paragraph{Correcting errors on $B\ni x$.} Having establish the previous step, we look at $A'\subseteq A$ of dimension $t+1$ that contains $x$, and note
that all erroneous affine subspaces $B\subseteq A'$ must contain $x$. As these constitute only $1/q$ fraction of the subspaces contained in $A'$, due to Lemma~\ref{lem:HSS}
they must all be erroneous. Next, we show it is possible to change $f(x)$ and make at least one of these $B$'s pass the test, which then by the second bullet in
Lemma~\ref{lem:HSS} guarantees that $f|_{A'}$ must be degree $d$. Indeed, $f|_{A'}$ has degree at most $(t+2)(q-1)$, and we can add to it a multiple of
$g(z) = 1_{z=x}$ to eliminate its highest degree monomial (as there is only one such monomial), say it is $f+g$. The function $(f+g)|_{A'}$ then has degree strictly smaller than
$(t+2)(q-1)$, hence by Lemma~\ref{lem:HSS} unless $(f+g)|_{A'}$ is degree at most $d$, it must be the case that more than $1/q$ fraction of the $B\subseteq A'$ of dimension $t$
fail the test. However, these can only still be subspaces containing $x$, which are at most $1/q$ fraction. Thus, $(f+g)|_{A'}$ has degree $d$, so that we showed that we may change
$f(x)$ and make $f|_{A'}$ degree $d$.

\section{Preliminaries}
\subsection{Relating different testers}
In this section, we provide several basic facts that will be used throughout the proof.
Below, the first bullet is~\cite{HSS}[Lemma 4.6], and the second bullet is a slight refinement which elaborates on when a given function $f$ may be tight for the first bullet.
Since the proof of the refinement is a small tweak on the original proof from~\cite{HSS},
we fully record it here.
\begin{lemma}\label{lem:HSS}
  Let $p$ be prime, $q\in\mathbb{N}$ be a power of $p$ and $d\in\mathbb{N}$ and set $t = \lceil{\frac{d+1}{q-q/p}}\rceil$. Suppose that
  $k\geq t$, and let $f\colon\mathbb{F}_q^{k+1}\to\mathbb{F}_q$. Then
  \begin{enumerate}
    \item If ${\sf deg}(f) > d$, then $\eps_{k,d}(f)\geq \frac{1}{q}$.
    \item If $d < {\sf deg}(f) < (k+1)(q-1)$, then $\eps_{k,d}(f) > \frac{1}{q}$.
  \end{enumerate}
\end{lemma}
\begin{proof}
  Let $f(x)$ has degree strictly larger than $d$. We shall think about restrictions to $k$-flats
  as taking a non-constant linear $L\colon \mathbb{F}_q^{k+1}\to\mathbb{F}_q$, and then considering
  $f|_{L=0}$. We shall use the notion of canonical monomials from~\cite{HSS}, which in our context reads:
  a monomial $M(x) = \prod\limits_{j\leq m}x_j^{e_j}$ is canonical if it appears in $f$,
  %, $e_1+\ldots+e_m = d$, and
  $q-q/p\leq e_1,\ldots,e_{m-1} \leq q-1$ and $e_m\leq q-1$. From~\cite{HSS}[Lemma 4.3] we may compose $f$ with an invertible affine
  linear transformation, and get to assume that $f$ has max-monomial of degree ${\sf ddeg}(f)$ which is canonical; clearly, once we prove the
  statement for this composition, the lemma immediately follows for the original function. We henceforth assume without loss of generality
  that this transformation is the identity.

  Let $M$ be a canonical max-monomial of $f$, and write $M(x) = \prod\limits_{j\leq m}x_j^{e_j}$
  and $m\leq k+1$. We consider two cases:
  \begin{itemize}
    \item {\bf Case 1: $m\leq k$.} In this case, we note that any linear transform $L$ that does not depend on the variables
    $x_1,\ldots,x_{m}$ preserves the degree of $f$, i.e. ${\sf deg}(f|_{L=0}) = {\sf deg}(f)$.

    For any other linear transformation $L$, it must depend on one of the variables $x_1,\ldots,x_m$, say without loss of generality
    it depends on $x_1$, and say $L(x) = a_1x_1 + L'(x_2,\ldots,x_{k}) + a_{k+1} x_{k+1}$ for $a_1\neq 0$. Denote $L_z(x) =
    a_1x_1 + L'(x_2,\ldots,x_{k}) + z x_{k+1}$, so that $L(x) = L_{a_{k+1}}(x)$.
    In this case, we may think of $f|_{L_z=0}$ as
    $f(-a_1^{-1}(L'(x_2,\ldots,x_{k+1}) + z x_{k+1}), x_2,\ldots,x_{k+1})$, and we show that there is $z\in\mathbb{F}_q$
    such that the degree of $f|_{L_{z} = 0}$ is greater than $d$. Thus, we conclude in this case that if $a_{k+1}$ was already
    this $z$ the degree of $f|_{L=0}$ would have been higher than $d$, and so $L$'s that depend on the variables $x_1,\ldots,x_m$
    we have that ${\sf deg}(f|_{L=0})>d$ with probability at least $1/q$. Together with the previous paragraph this establishes
    both items of the lemma in this case, and we next show the existence of this $z$.

    The idea is to look at $f(-a_1^{-1}(L'(x_2,\ldots,x_{k+1}) + z x_{k+1}), x_2,\ldots,x_{k+1})$, and more specifically at
    the coefficient of the monomial $M'(x) = \prod\limits_{1<j\leq m}x_j^{e_j}\cdot {x_{k+1}}^e_1$. The max-monomial $M$
    from $f$ would give us this monomial with coefficient $-a_1^{-1} z^{e_1}$, and since $M$ was max-monomial any other
    monomials will be able to contribute only $z$'s with lower power. Hence, the coefficient of $M'$ is some non-zero polynomial
    in $z$ of degree at most $e_1\leq q-1$, and hence we may choose $z$ for which it is non-zero.

    \item {\bf Case 2: $m=k+1$.} Suppose $e_1\geq e_2\geq\ldots\geq e_m$.
    Here we consider two subcases.
    \begin{itemize}
      \item First, consider the case that  $e_1+\ldots+e_m < (k+1)(q-1)$ (which is the only case
    we need for the second bullet in the lemma), so that $e_m < q-1$. Choose a non-constant
    linear transformation $L(x_1,\ldots,x_{k+1}) = \sum\limits_{i=1}^{k+1} a_i x_i + c$ randomly, and note that the probability that $a_{k+1}$ is $0$ is strictly
    smaller than $1/q$ (indeed, the distribution of $(a_1,\ldots,a_{k+1})$ is over non-zero vectors).
    We shall focus on $L$'s such that $a_{k+1}\neq 0$. Let $L_z(x_1,\ldots,x_{k+1})= \sum\limits_{i=1}^{k+1} a_i x_i + z$, and we argue that for each $L$,
    there are at least $2$ values of $z$ for which $f|_{L_z = 0}$ has degree greater than $d$. Indeed, the argument is exactly the same as before, except that
    we look at the monomial $M' = \prod\limits_{i=1}^{k} x_i$, and note that from $M$ we have a contribution $a_{k+1}^{-1} z^{e_{k+1}}$, and as $M$ is a max-monomial
    all other contributions are lower degree in $z$. Hence, choosing $z$ at random the probability this coefficient is non-zero is at least
    $\frac{q-e_{k+1}}{q} \geq \frac{2}{q}$, and in this case the degree of $f|_{L_z = 0}$ is at least $e_1+\ldots+e_k \geq k(q-q/p)\geq d+1$.

    Thus, we get that the probability that $f|_{L=0}$ has degree greater than $d$ is
    \[
    >\frac{q-1}{q}\cdot\frac{2}{q} \geq \frac{1}{q},
    \]
    where the first factor comes from the event $a_1\neq 0$, and the second factor comes from the event that $a_{k+1}$ is one of the two $z$'s which keeps
    the degree of $f|_{L_z=0}$ high.
    \item Next, consider the case that $e_1+\ldots+e_m \geq (k+1)(q-1) $. This case is similar to case $1$. A non-constant affine transformation $L$
    either has the form $L(x) = \sum\limits_{i=1}^{k+1} a_ix_i +c$, and letting $L_z(x) = \sum\limits_{i=1}^{k+1} a_ix_i + z$, we show that there
    is $z\in\mathbb{F}_q$ such that ${\sf deg}(f|_{L_z = 0})\geq d+1$. Indeed, suppose without loss of generality that $a_{k+1}\neq 0$, then
    the coefficient of the monomial $\sum\limits_{i=1}^{k} a_ix_i$ in $f|_{L_z = 0}$ is a non-zero polynomial in $z$ of degree at most $q-1$, hence
    there is $z$ for which this coefficient is non-zero, and hence $f|_{L_z = 0}$ has degree $k(q-1)\geq d+1$.

    As $c=z$, which happens with probability $1/q$, we have that the degree of $f|_{L}$ is strictly larger than $d$.
    \qedhere
    \end{itemize}
  \end{itemize}
\end{proof}

\begin{lemma}\label{lem:relate}
  Let $p$ be a prime, $q$ be a power of $p$, $d\in\mathbb{N}$ and let $t = \lceil{\frac{d+1}{q-q/p}}\rceil$. Suppose that
  $k\geq t$, then for all $f\colon \mathbb{F}_q^n\to\mathbb{F}_q$ we have $\eps_{k+1,d}(f)\leq q \eps_{k,d}(f)$.
\end{lemma}
\begin{proof}
  Let $B\subseteq\mathbb{F}_q^n$ be a uniform $k+1$ dimensional flat, and let $A\subseteq B$ be a uniform $k$-dimensional flat.
  Then
  \[
  \eps_{k,d}(f) = \Prob{A,B}{{\sf deg}(f|_{A})>d} = \Prob{A,B}{{\sf deg}(f|_{B})>d}\cProb{A,B}{{\sf deg}(f|_{B})>d}{{\sf deg}(f|_{A})>d}.
  \]
  The first probability on the right hand side is $\eps_{k+1,d}(f)$, and the second probability is at least $1/q$ by Lemma~\ref{lem:HSS}.
\end{proof}

\subsection{Expansion and pseudo-randomness}
Denote by $V_q(k,\ell)$ the set of dimension $\ell$ affine subspaces in $\mathbb{F}_q^k$; we often omit the subscript $q$ when it is clear from context.
In this section, we discuss expansion in the affine Grassmann graph over $V_q(k,\ell)$. Similarly to the works~\cite{DKKMS2,KMS2}, we too consider certain structures of sets that forbid strong expansion properties, but
in our case there are additional types of structures (due to the fact we are working in the affine case).
\begin{definition}\label{def:pseudo_random}
  Let $S\subseteq V_q(k,\ell)$, and let $\xi\in [0,1]$.
  \begin{enumerate}
    \item We say $S$ is $\xi$-pseudo-random with respect to hyperplanes, if for each affine hyperplane $W$ we have that
    \[
    \mu(S_W)\defeq \Prob{A\in V(W,\ell)}{A\in S}\leq \xi.
    \]
    \item We say $S$ is $\xi$-pseudo-random with respect to hyperplanes on its linear part, if for each hyperplane $W$ we have that
    \[
    \mu(S_{W,{\sf lin}})\defeq \Prob{A\in V(W,\ell),x\not\in W}{x+A\in S}\leq \xi.
    \]
    \item We say $S$ is $\xi$-pseudo-random with respect to points on its linear part if for each point $y$, we have that
    \[
    \mu(S_{x,{\sf lin}})\defeq\cProb{A = x+A'\in {\sf AffGrass}(k, \ell)}{A'\ni y}{A\in S}\leq \xi.
    \]
    \item We say $S$ is $\xi$-pseudo-random with respect to points, if for each point $x$ we have that
    \[
    \mu(S_x)\defeq \cProb{A \in {\sf AffGrass}(k, \ell)}{x\in A}{A\in S}\leq \xi.
    \]
  \end{enumerate}
\end{definition}
In~\cite{DKKMS2,KMS2} it is proved that pseudo-random sets have strong expansion properties in the Grassmann graph.
Here, we require a similar statement. Consider $W$ a $k$-dimensional affine space over
$\mathbb{F}_p$, and consider the random walk on ${\sf AffGras}(W,\ell)$ as described in the introduction.

%The eigenvalues of $T$ are $\lambda_j = p^{-j} + o(1)$ for $j=0,1,\ldots,k-1$, and we define the expansion of a set $S$
%as
%\[
%1-\Phi(S) =\frac{1}{\mu(S)}\Prob{K,K'\sim T}{K,K'\in S}.
%\]
\begin{thm}\label{thm:expansion}
  Let $\xi>0$ and let $q\in\mathbb{N}$ be a prime power.
  Suppose that $S\subseteq V_q(W,\ell)$ is a set with:
  \begin{enumerate}
    \item $\mu(S)\leq \xi$;
    \item $S$ is $\xi$-pseudo-random with respect to hyperplanes, with respect to hyperplanes on its linear part, as well as with respect to points on its linear part;
    \item $1-\Phi(S)\geq \frac{1}{q}$.
  \end{enumerate}
    Then there exists a point $x\in\mathbb{F}_q^n$ such that $\mu(S_x)\geq 1-q^{2}(867\xi^{1/4} + q^{-\ell})$
\end{thm}
The proof of Theorem~\ref{thm:expansion} proceeds similarly to the proof presented in~\cite{DKKMS2} for the degree $1$ case, however
we use some of the machinery of~\cite{KMS2} to simplify the presentation. The proof is deferred to the appendix.

\subsection{Expansion and sharp thresholds}
\begin{definition}
  For $h\geq \ell$ and $S\subseteq V(k,\ell)$, we define
  \[
  S\uparrow^{h} = \sett{L}{{\sf dim}(L) = h, \exists K\in S, K\subseteq L}.
  \]
  When $h = \ell+1$, we omit the superscript $h$.
\end{definition}
The following lemma is very similar to~\cite[Proposition III.3.4.]{KLLM}.
\begin{lemma}\label{lem:sharp_thresh}
  $\mu(S\uparrow)\geq\frac{\mu(S)}{1-\Phi(S)}$.
\end{lemma}
\begin{proof}
  For $B\in V(k,\ell)$, denote by $T\uparrow B$ the uniform distribution over subspaces $B'\in V(k,\ell+1)$ containing $B$,
  and for $B'\in V(k,\ell+1)$ denote by $T\downarrow B'$ the uniform distribution over $B\in V(k,\ell)$ contained in $B'$. We consider
  real-valued functions over $V(k,\ell)$, $V(k,\ell+1)$ and view $T\uparrow$, $T\downarrow$ as operators,
  $T\downarrow\colon L_2(V(K,\ell))\to L_2(V(K,\ell+1))$, $T\uparrow\colon L_2(V(K,\ell+1))\to L_2(V(K,\ell))$ defined as
  \[
  T\downarrow f(B') = \Expect{B\sim T\downarrow B'}{f(B)},
  \qquad\qquad
  T\uparrow g(B) = \Expect{B'\sim T\uparrow B}{g(B')},
  \]
  for $f\colon V(k,\ell)\to\mathbb{R}$, $g\colon V(k,\ell+1)\to\mathbb{R}$. We note that $T\downarrow$ is the adjoint of $T\uparrow$.

  Fix $S$, and let $f = 1_S$, $g = 1_{S\uparrow}$. Then
  \[
  \mu(S)
  = \Expect{B\in V(k,\ell)}{f(B)}
  = \Expect{B'\in V(k,\ell+1)}{\Expect{B\sim T\downarrow B'}{ f(B)}}
  = \Expect{B'\in V(k,\ell+1)}{g(B')T\downarrow f(B')}
  = \inner{g}{T\downarrow f}.
  \]
  Thus, using Cauchy-Schwarz
  \[
  \mu(S)^2
  \leq \norm{g}_2^2 \norm{T\downarrow f}_2^2
  = \mu(S\uparrow) \inner{T\downarrow f}{T\downarrow f}
  = \mu(S\uparrow) \inner{f}{T\uparrow T\downarrow f}.
  \]
  Thus,
  \[
  \mu(S\uparrow)\geq \frac{\mu(S)}{\frac{1}{\mu(S)}\inner{f}{T\uparrow T\downarrow f}}.
  \]
  We note that for $B$, the distribution of $\tilde{B}\sim T\uparrow T\downarrow B$ is distributed according to the random walk of
  ${\sf AffGras}(\mathbb{F}_q^k,\ell)$, hence $\frac{1}{\mu(S)}\inner{f}{T\uparrow T\downarrow f} = 1-\Phi(S)$, finishing the proof.
\end{proof}

\section{Testing Reed-Muller codes: proof of Theorem~\ref{thm:main}}\label{sec:pf_of_main}
In this section, we present the formal proof of Theorem~\ref{thm:main}. For a high level description of our proof strategy we defer the reader
to Section~\ref{sec:techniques}.

\subsection{Step 1: locating a potential error}
Fix $f$ as in the statement of the theorem, and denote
\[
S = \sett{A}{{\sf dim}(A) = t, {\sf deg}(f|_{A})>d}.
\]
We note that $\eps_{t,d}(f) = \mu(S) \defeq \eps$, and that $\eps_{t+1,d}(f) = \mu(S\uparrow)$. Throughout, we will assume that
$\eps\leq q^{-M}$ for a sufficiently large (but absolute) constant $M$. We will also assume that $t\geq M$, otherwise the result follows
from~\cite{KR}.

Let $\eps'=\max(\eps,q^{-d})$.
Our aim in this section is to prove the following proposition.
\begin{proposition}\label{prop:find_zoom_in}
  There exists $x^{\star}\in\mathbb{F}_q^n$ such that
  $\mu(S_{x^{\star}})\geq 1-C(q) \eps'^{1/4}$ for $C(q) = 2000q^2$.
\end{proposition}
We will prove this proposition using Theorem~\ref{thm:expansion}, and towards
this end we first show that the conditions of Theorem~\ref{thm:expansion} hold.

\begin{claim}\label{claim:reduce_to_expansion}
  $1-\Phi(S)\geq \frac{1}{q}$.
\end{claim}
\begin{proof}
  By Lemma~\ref{lem:relate} we have $\mu(S\uparrow)\leq q\mu(S)$ and by Lemma~\ref{lem:sharp_thresh}
  we have $\mu(S\uparrow)\geq \frac{\mu(S)}{1-\Phi(S)}$. Combining these two inequalities gives
  the statement of the claim.
\end{proof}

To apply Theorem~\ref{thm:expansion}, we first argue that $S$ is pseudo-random.
\begin{claim}\label{claim:pseudo_zoom_out}
  The set $S$ is $2q\mu(S)$ pseudo-random with respect to zoom-outs, also with respect to the linear part.
\end{claim}
\begin{proof}
  Let $W\subseteq \mathbb{F}_q^n$ be a hyperplane (either affine or not), and sample a $(t+1)$-flat uniformly
  $A\subseteq\mathbb{F}_q^n$. We note that the probability that $W\cap A$ has dimension $t$ is $1-q^{-(t+1)}$.
  To see that, we may think of $W$ as being defined by an equation $\inner{x}{h} = c$ for some non-zero vector $h$ and $c\in\mathbb{F}_q$,
  and $A$ as being defined by a collection of linearly independent equations $\inner{x}{h_i} = c_i$ for $i=1,\ldots n-t-1$. Whenever
  $h\not\in {\sf span}(h_1,\ldots,h_{n-t-1})$, $A\cap W$ has dimension $t$. Conditioned on this event, $A\cap W$ is a uniform
  $t$-flat in $W$, and so $A\cap W\in S$ with probability $\mu(S_W)$. Also, if $A\cap W\in S$ then $A\in S\uparrow$, so we get that
  \[
    \mu(S\uparrow) = \Prob{}{A\in S\uparrow}\geq (1-q^{-(t+1)})\mu(S_W)\geq \frac{\mu(S_W)}{2}.
  \]
  Thus, $\mu(S_W)\leq 2 \mu(S\uparrow)\leq 2q\mu(S)$.
\end{proof}

Next, we prove that $S$ is pseudo-random with respect to zoom-ins on its linear part. This argument is more involved, and requires a
bootstrapping-style argument as described in the proof overview; namely, we show that if there are very little errors on specific type of
subspaces, then there must be no errors at all on these type of subspaces.
\begin{claim}\label{claim:pseduo_zoom_in_linear}
  The set $S$ is $q^{200-M}$ pseudo-random with respect to zoom-in on its linear part.
\end{claim}
\begin{proof}
  Suppose otherwise, then there is $z\in\mathbb{F}_q^n$ such that
  \[
  \sett{x+A\in S}{A\subseteq \mathbb{F}_q^n\text{ linear subspace of dimension $t$}, z\in A},
  \]
  has fractional size $\alpha > q^{200-M}$ inside $P_{z,t} = \sett{x+A}{A\subseteq \mathbb{F}_q^n\text{ linear subspace of dimension $t$}, z\in A}$.
  Clearly,
  \[
  S_z' = \sett{x+A\in S\uparrow^{t+100}}{A\subseteq \mathbb{F}_q^n\text{ linear subspace of dimension $t+100$}, z\in A}
  \]
  also has at least $\alpha$ fractional size inside $P_{z,t+100}$. Take $x+A$ a $(t+100)$-flat uniformly, consider the event it is in $S_z'$ , and take a $t$-flat $B = x' + A'\subseteq x+A$ uniformly; note that
  \begin{align*}
  \cProb{\substack{x+A\\ B = x'+A'\subseteq x+A}}{x+A\in S_z', z\not\in A'}{{\sf deg}(f|_{B}) > d}
  &\leq \frac{\cProb{x+A, B = x'+A'}{z\in A, z\not\in A'}{{\sf deg}(f|_{B}) > d}}{\cProb{}{z\in A,z\not\in A'}{x+A\in S_z'}}\\
  &\leq \frac{\cProb{B = x'+A'}{z\not\in A'}{{\sf deg}(f|_{B}) > d}}
  {\cProb{}{z\in A}{x+A\in S_z'}\cProb{}{x+A\in S_z'}{z\not\in A'}}.
  \end{align*}
  The numerator is at most $\eps$, and the denominator is at least $\alpha/2$, so we get this probability is at most $\frac{2}{\alpha}\eps$. Thus,
  there exists $x+A\in S_z'$ such that conditioned on it this probability is at most $\frac{2}{\alpha}\eps$, and we fix it henceforth.

  We now work over the affine $t$-dimensional Grassmann graph over $x+A$.
  Consider $t$-flats $B = x'+A'\subseteq x+A$ conditioned on $A'$ not containing $z$, and let
  \[
  \mathcal{B} = \sett{B = x'+A'\in {\sf AffGras}(x+A, t)}{z\not\in A', {\sf deg}(f|_{B}) > d}.
  \]
  We argue that $\mathcal{B}$ must be empty; suppose towards contradiction otherwise.
  Let $W = y+\tilde{W}\subseteq x+A$ be randomly chosen where $\tilde{W}$ is uniformly chosen linear subspace of dimension $t+40$ not containing $z$,
  and $y\in x+A$ is uniformly chosen. Denote
  \[
  \mathcal{B}_W = \sett{B\in\mathcal{B}}{B\subseteq W}.
  \]
  We denote by $\mu_W$ the uniform measure over ${\sf AffGras}(W, t)$. We argue that for all $W$, if $\mu_W(\mathcal{B}_W)\leq q^{-100}$, then
  $\mu_W(\mathcal{B}_W) = 0$. Indeed, if $\mu_W(\mathcal{B}_W)\leq q^{-100}$ then $\mathcal{B}_W$ is $q^{-60}$ pseudo-random with respect to zoom-ins
  (also with respect to its linear part), as those have measure at least $q^{-40}$. Also, by an argument as in Claim~\ref{claim:pseudo_zoom_out}
  we have that $\mathcal{B}_W$ is $q^{-50}$ pseudo-random with respect to zoom out (also with respect to its linear part).
  Finally, by an argument as in Claim~\ref{claim:reduce_to_expansion} we have
  \[
  1-\Phi_W(\mathcal{B}_W)\geq \frac{1}{q},
  \]
  where $\Phi_W$ is expansion with respect to ${\sf AffGras}(W, t)$.
  This is now a contradiction to Theorem~\ref{thm:expansion}. We thus conclude that either $\mu_W(\mathcal{B}_W) = 0$ or
  $\mu_W(\mathcal{B}_W)\geq q^{-100}$; as $\mathcal{B}$ is non-empty (by our assumption), we may find
  $W$ such that $\mu_W(\mathcal{B}_W)\geq q^{-100}$, and we fix such one.

  Next, we take a uniform $Y = u + \tilde{Y}\subseteq x+A$ of dimension $t+99$ conditioned on $z\not\in \tilde{Y}$,
  sample a $(t+60)$-flat $A_2\subseteq Y$, and consider $A_2\cap W$. We may think of $W$ as being defined by a system of
  $60$ independent linear equations $\inner{h_1}{x} = c_1,\ldots,\inner{h_{60}}{x} = c_{60}$ over $x+A$, and $A_2$ as being defined
  by a set of $39$ linear equations $\inner{h_1'}{x} = c_1',\ldots,\inner{h_{39}'}{x} = c_{39}'$ where $h_1,\ldots,h_{39}'$ are random
  linearly independent. Note that the probability that ${\sf span}(h_1,\ldots,h_{60},h_1',\ldots,h_{39}')$ has dimension $99$ is
  at least
  \[
  \prod\limits_{j=0}^{38}\frac{q^{99} - q^{60+j}}{q^{99}}\geq e^{-2\sum\limits_{j=1}^{\infty} q^{-j}}\geq e^{-4/q},
  \]
  in which case the distribution of $A_2\cap W$ is uniform from ${\sf AffGras}(W,t)$. Hence, $A_2\cap W$ is in $\mathcal{B}_W$ with
  probability at least $e^{-4/q} q^{-100}$. In this case, we have that $A_2\in \mathcal{B}_Y \uparrow^{t+60}$ where upper shadow is taken
  with respect to ${\sf AffGras}(Y, t)$. Thus, we get that
  \[
  \Expect{Y}{\mu_Y(\mathcal{B}_Y \uparrow^{t+60})}\geq e^{-4/q} \mu_{W}(\mathcal{B}_W)\geq e^{-4/q} q^{-100}.
  \]
  However, by Lemma~\ref{lem:relate} for each $Y$ we have that
  \[
  \mu_Y(\mathcal{B}_Y \uparrow^{t+60})\leq q^{60} \mu_Y(\mathcal{B}_Y),
  \]
  and plugging that in above we get that
  \[
  \Expect{Y}{\mu_Y(\mathcal{B}_Y)} \geq e^{-4/q} q^{-160}.
  \]
  Finally, the left hand side is at most the probability that $f|_{x'+A'}$ has degree $>d$ when $x'+A'\subseteq x+A$ is
  a random $t$-flat conditioned on $A\not\ni z$, hence at most $\frac{2}{\alpha}\eps$ by the choice of $x+A$. Overall, we get that
  \[
  \alpha\leq e^{4/q} q^{160} 2\eps\leq q^{200-M},
  \]
  and contradiction. This contradiction implies that $\mathcal{B}$ is empty, and we quickly finish the argument now.

  Let us look at $x+A$; as $f|_{x+A}$ has degree larger than $d$, we may find a $(t+1)$-flat $B = x'+\tilde{B}\subseteq x+A$ such that
  $f|_B$ has degree larger than $d$. Sample a $t$-flat $x''+B'\subseteq B$ uniformly. By the above, if $z\not\in B'$, we have that
  $f|_{x''+B'}$ has degree $d$. Note that the probability that $z\in B'$ is at most
  \[
  \frac{q^{t}-1}{q^{t+1}-1} < \frac{1}{q},
  \]
  so we get that for less than $1/q$ fraction of the $t$-flats $x''+B'\subseteq B$ we have that ${\sf deg}(f|_{x''+B'})>d$. This contradicts
  Lemma~\ref{lem:relate}.
\end{proof}

We can now prove Proposition~\ref{prop:find_zoom_in}.
\begin{proof}[Proof of Proposition~\ref{prop:find_zoom_in}]
From Claims~\ref{claim:reduce_to_expansion},~\ref{claim:pseudo_zoom_out},~\ref{claim:pseduo_zoom_in_linear} we have that the conditions
of Theorem~\ref{thm:expansion} hold, and hence we may find $x^{\star}\in\mathbb{F}_p^n$ such that $\mu(S_{x^{\star}})\geq 1-C(q) \eps'^{1/4}$,
for $C(q) = 2000q^2$.
\end{proof}

\subsection{Step 2: correcting the value on $x^{\star}$}\label{sec:correct_val}
The goal of this section is to prove the following proposition.
\begin{proposition}\label{prop:fix_x_star}
  There exists $c\in\mathbb{F}_q$ such that changing the value of $f(x^{\star})$ to $c$, we have that
  \[
  \cProb{A''\text{ $t$-flat}}{x^{\star}\in A''}{{\sf deg}(f|_{A''})\leq d}\geq \frac{1}{2q}.
  \]
\end{proposition}
%Next, we all $t$-flats containing $x^{\star}$. We show that we may find a setting for $f(x^{\star})$, such
%that the $t$-flat test would pass on at least $\frac{1}{2p}$ of them, which is significantly larger than the
%fraction of them that currently pass (which is upper bounded by $C(p)\eps$).
The rest of this section is devoted to proving Proposition~\ref{prop:fix_x_star}. Take a uniform $(t+100)$-flat $A$ containing $x^{\star}$, and let
\[
  \mathcal{B}_A = \sett{B \in {\sf AffGras}(A, t)}{x^{\star}\not\in B, {\sf deg}(f|_{B}) > d},
\]
then $\Expect{A}{\mu_A(\mathcal{B}_A)}\leq O(\eps)$, so with probability at least $1/2$ over $A$ we have
that $\mu_A(\mathcal{B}_A)\leq O(\eps)$.

Take a $(t+40)$ flat $W\subseteq A$ randomly not containing $x^{\star}$, and let
\[
  \mathcal{B}_W = \sett{B\in {\sf AffGras}(W,t)}{ B\in \mathcal{B}_A}.
\]
We argue that for each $W$, either $\mu_W(\mathcal{B}_W) = 0$ or $\mu_W(\mathcal{B}_W)\geq q^{-100}$. Otherwise,
$0<\mu_W(\mathcal{B}_W)<q^{-100}$. Therefore, $\mathcal{B}_W$ is $q^{-60}$ pseudo-random with respect to zoom ins
(also with respect to their linear part), and from an argument as in Claim~\ref{claim:pseudo_zoom_out} we have that
$\mathcal{B}_W$ is $q^{-98}$ pseudo-random with respect to zoom-outs (as well as their linear parts). Finally,
as in the argument in Claim~\ref{claim:reduce_to_expansion} we have $1-\Phi_W(\mathcal{B}_W)\geq 1/q$, so we get
a contradiction to Theorem~\ref{thm:expansion}.

\begin{claim}\label{claim:is_empty}
$\mathcal{B}_A = \emptyset$.
\end{claim}
\begin{proof}
Otherwise, we may find $W$ such that $\mu_W(\mathcal{B}_W)\geq q^{-100}$.
The argument is similar to the end of the argument in Claim~\ref{claim:pseduo_zoom_in_linear}. Take a $(t+99)$ flat $Y\subseteq A$ randomly
not containing $x^{\star}$, and take a $(t+60)$-flat $A_2\subseteq Y$ randomly. Then $A_2\cap W$ has dimension $t$ with probability at least
$e^{-4/q}$, and then its distribution is uniform in ${\sf AffGras}(W,t)$. Thus, it is in $\mathcal{B}_W$ with probability at least $q^{-100}$.
Therefore, we get that
\[
\Expect{Y}{\mu_Y(\mathcal{B}_Y\uparrow^{t+60})}\geq \Prob{Y, A_2}{A_2\cap W\in \mathcal{B}_W}\geq e^{-4/q}\cdot q^{-100}.
\]
On the other hand, by Lemma~\ref{lem:relate}
\[
\Expect{Y}{\mu_Y(\mathcal{B}_Y\uparrow^{t+60})}
\leq q^{60}\Expect{Y}{\mu_Y(\mathcal{B}_Y)}
\leq q^{60}2\mu_A(\mathcal{B})
\leq 2q^{60}C(q)\eps'^{1/4}.
\]
Combining the two, we get that
\[
q^{-M/4}\geq \eps'^{1/4}\geq \frac{1}{C(q)q^{160}},
\]
which is a contradiction for large enough $M$.
\end{proof}

We are now ready to prove Proposition~\ref{prop:fix_x_star}.
\begin{proof}[Proof of Proposition~\ref{prop:fix_x_star}]
Take any $(t+1)$-flat $A'\subseteq A$ containing $x^{\star}$, and define $g = f|_{A'}$. Consider the polynomial $M(x) = 1_{x\neq x^{\star}}$ on $A'$.
Note that $M$ has degree $(t+2)(q-1)$, so we may find a constant $c\in\mathbb{F}_p$ such that $g' = g+cM$ has degree strictly smaller than $(t+2)(q-1)$.
We claim that ${\sf deg}(g') \leq d$.
Otherwise, from Lemma~\ref{lem:HSS} the fraction of $t$-flats $B\subseteq A'$ such that $g'|_{B}$ has degree greater than $d$ is strictly larger than $1/q$.
As the fraction of $B$'s that contain $x^{\star}$ is exactly $1/q$, it follows that there is $B\subseteq A'$ not containing $x^{\star}$ such that
${\sf deg}(g'|_{B}) > d$. But for such $B$'s we have ${\sf deg}(g'|_{B}) = {\sf deg}(f|_B)\leq d$, and contradiction. Thus, ${\sf deg}(g')\leq d$.
Stated otherwise, we may change the value of $f(x^{\star})$ and make the degree of $f|_{A'}$ at most $d$. In particular, we get that for each $t$-flat
$A''\subseteq A$ we may change $f(x^{\star})$ and make the degree of $f|_{A''}$ at most $d$.

Sampling $A$ a $(t+100)$ flat containing $x^{\star}$ randomly and then a $t$-flat $A''\subseteq A$ containing $x^{\star}$, we get that with probability at least $1/2$ we may change $f(x^{\star})$ and make the degree of $f|_{A''}$ at most $d$. Thus, taking the plurality vote we may choose $f(x^{\star})$ that appeases at least $\frac{1}{2q}$ of the $t$-flats
containing $x^{\star}$.
\end{proof}

\subsection{Fixing the error and iterating}\label{sec:iterate}
\begin{proposition}\label{prop:single_iteration}
  We may find $x\in\mathbb{F}_q^n$ and function $f'$ which is identical to $f$ at all points
  except at $x$, such that
  \[
  \eps_{t,d}(f')\leq \eps_{t,d}(f) - q^{t-d} \frac{1}{4q}
  \]
\end{proposition}
\begin{proof}
  Using Proposition~\ref{prop:find_zoom_in} we find $x^{\star}$ such that $\mu(S_{x^{\star}})\geq 1-C(q)\eps'$, and using Proposition~\ref{prop:fix_x_star}
  we find $c\in\mathbb{F}_q$ such that taking $f'$ to be identical to $f$ at all points except at $x^{\star}$ where it is equal to $c$, we have that
  $f'$ passes at least $\frac{1}{2q}$ fraction of the tests containing $x^{\star}$. We compare the probability that $f$ and $f'$ pass the $t$-flat test.
  Sample a $t$-flat $A$. Clearly, if $A$ does not contain $x^{\star}$ they perform the same; otherwise, $f$ passes with probability at most $C(q)\eps'$,
  and $f'$ passes with probability at least $\frac{1}{2q}$. As the probability that $x^{\star}\in A$ is $q^{t-n}$, we get that
  \[
  \eps_{t,d}(f')\leq \eps_{t,d}(f) - q^{t-d}\left(\left(1- O(\eps')\right) - \left(1-\frac{1}{2q}\right)\right)\leq \eps{t,d}(f) - q^{t-d} \frac{1}{4q}.\qedhere
  \]
\end{proof}
From Proposition~\ref{prop:single_iteration} we get that as long as $\eps_{t,d}(f)>0$, we may find a point $x$ and change $f(x)$ so as to decrease
$\eps_{t,d}(f)$ by at least $q^{t-d}\frac{1}{4q}$. Thus, after at most $\frac{\eps_{t,d}(f)}{q^{t-d}/4q}$ invocations of the proposition we will end
up with a function that passes the test with probability $1$, which by the choice of $t$ implies we will end up with a degree $d$ function. We therefore get that
\[
\delta_d(f)q^n\leq \frac{\eps_{t,d}(f)}{q^{t-d}/4q},
\]
hence $\delta_d(f)\leq 4q^{1-t}\eps_{t,d}(f)$.
\qed

\section{Lifted affine invariant codes: proof of Theorem~\ref{thm:main2}}
In this section, we argue that the method above used to prove optimal testing for Reed-Muller codes applies to lifted affine invariant codes
as well, thereby proving Theorem~\ref{thm:main2}. Towards this end, it turns out that the only part that has to be adjusted are Lemmas~\ref{lem:HSS} and~\ref{lem:relate}.
Thus, we begin by proving them for affine invariant codes and then quickly explain how the rest of the proof proceeds.

\subsection{Facts about affine invariant codes}
\begin{definition}\label{def:deg_domination}
  Let $m,n\in\mathbb{N}$, let $p$ be prime and write $m = \sum\limits_{i=0}^{r} m_i p^{i}$, $n = \sum\limits_{i=0}^{r} n_i p^i$ the base $p$ expansion of $m$ and $n$.
  We say $m$ dominates $n$ with respect to the $p$-base expansion if $m_i\geq n_i$ for all $i$.
\end{definition}

For a polynomial $f$, we denote by ${\sf supp}(f)$ the collection of monomials in $f$ that have a non-zero coefficient. Also, for a set of functions $\mathcal{B}$,
we denote by ${\sf supp}(\mathcal{B})$ the set of monomials that appear in at least one of these functions. Lastly, we will use the fact
that the support of an affine invariant set is affine invariant.

\begin{lemma}\label{lem:spreading}[Monomial spreading~\cite[Lemma 4.6]{KaufmanSudan}]
  Suppose that $\mathcal{B}$ is affine invariant, and let $M = x_1^{d_1+e}x_2^{d_2}x_3^{d_3}\cdots x_t^{d_t}$ and $M' = x_{1}^{d_1}x_2^{d_2+e}\cdots x_t^{d_t}$ be monomials
  such that $d_1+e$ dominates $e$. If $M\in{\sf supp}(\mathcal{B})$, then $M'\in{\sf supp}(\mathcal{B})$.
\end{lemma}

Finally, we will use the following characterization of affine invariant codes, saying that they can be characterized by a monomial basis.
\begin{lemma}\label{lem:mon_basis}[~\cite[Lemma 4.2]{KaufmanSudan}]
  If $\mathcal{B}$ is an affine invariant linear code, then $\mathcal{B} = {\sf span}({\sf supp}(\mathcal{B}))$.
\end{lemma}
Thus, to show that $f\not\in\mathcal{B}$ it suffices to show that the support of $f$ contains a monomial not in $\mathcal{B}$.

\subsection{The relation lemma}
We begin by adapting Lemma~\ref{lem:HSS} to our case, following the argument in~\cite{HRZS}.
\begin{lemma}\label{lem:HSS2}
  Let $p$ be prime, $q\in\mathbb{N}$ be a power of $p$ and let $t\in\mathbb{N}$.
  Let $\mathcal{B}\subseteq \set{g\colon\mathbb{F}_q^t\to\mathbb{F}_q}$ be an affine invariant code,
  and denote $\mathcal{F} = {\sf Lift}_{k+1}(\mathcal{B})$. Suppose that $k\geq t$, and let $f\colon\mathbb{F}_q^{k+1}\to\mathbb{F}_q$
  be such that $f\not\in\mathcal{F}$. Then
  \begin{enumerate}
    \item $\eps_{k}(f)\geq \frac{1}{q}$.
    \item If $\eps_{k}(f) = \frac{1}{q}$, and the set $\mathcal{H}$ of hyperplanes $H$ for which $f|_{H}\not\in\mathcal{F}$
    is of the form
    \[
    \mathcal{H} = \sett{H\subseteq \mathbb{F}_q^{k+1}}{x^{\star}\in H}
    \]
    for some $x^{\star}\in\mathbb{F}_q^{k+1}$, then there exists $g\in\mathcal{F}$ that agrees with $f$ on all points except on
    $x^{\star}$.
  \end{enumerate}
\end{lemma}
\begin{proof}
Assume without loss of generality that $x^{\star} = 0$.
We will closely follow the argument in~\cite[Lemma 5.3]{HRZS} (we note that out assumption about $x^{\star}$ does not conflict with
the assumption therein that $T$ is the identity), which already establishes the first bullet. Our goal henceforth will be
to establish the second bullet.

A hyperplane therein is indexed by $\vec{\alpha} = (\alpha_0,\alpha_1,\ldots,\alpha_{k+1})$, which encodes the hyperplane
\[
H = \sett{x}{\alpha_0 + \sum\limits_{i=1}^{k+1} \alpha_i x_i = 0}.
\]
For each hyperplane, let $c_{\alpha}$ be the smallest $i\geq 1$ such that $\alpha_i\neq 0$. The argument in~\cite{HRZS} proceeds as follows:
\begin{enumerate}
  \item If $c_{\alpha}>t$, the authors show that $H\in\mathcal{H}$ given that $\alpha_0 = 0$. Hence, among the hyperplanes for which $c_{\alpha}>t$,
  at least $1/q$ of them lie in $\mathcal{H}$.
  \item If $1\leq c_{\alpha}\leq t$, then one may alter $\alpha_n$ and cause $H$ to be in $\mathcal{H}$. Hence, at least $1/q$ fraction of these hyperplanes
  are in $\mathcal{H}$. We note that if for some $\alpha$, there were at least $2$ ways of choosing $\alpha_n$ so that $H\in\mathcal{H}$, then we would get
  that the fraction of hyperplanes in $\mathcal{H}$ from this case is strictly greater than $1/q$. Thus, since we assume that $\eps_{k}(f) = 1/q$, there is
  precisely one way of choosing $\alpha$ so that $H\in\mathcal{H}$.
\end{enumerate}

We now look more closely at their analysis in the second case, starting from $c_{\alpha}=1$. Consider as there
 \[
    B(x_1,\ldots,x_n) = \left(x_{1} - \sum\limits_{1<j\leq n} \frac{\alpha_j}{\alpha_1} x_j - \frac{\alpha_0}{\alpha_1}, x_2,\ldots,x_n\right),
    \]
    and $B'\colon \mathbb{F}_q^t\to\mathbb{F}_q^t$ defined as
    \[
    B'(x_1,\ldots,x_t) = \left(x_{1} - \sum\limits_{1<j\leq t} \frac{\alpha_j}{\alpha_1} x_j - \frac{\alpha_0}{\alpha_1}, x_{2},\ldots,x_t\right).
    \]
    Note that $(f\circ B)_{x_{t+1} = 0,\ldots,x_n = 0} = f|_{x_{t+1}=0,\ldots,x_n=0} \circ B' \not \in\mathcal{B}$, as $f|_{x_{t+1}=0,\ldots,x_n=0}\not\in\mathcal{B}$
    and $\mathcal{B}$ is affine invariant. Thus, there is a monomial $M$ in the support of $f\circ B$ that is not in ${\sf supp}(\mathcal{B})$, say
    \[
    M = \prod\limits_{i=1}^{t} x_i^{d_i}.
    \]
    Let $\alpha(z) = (\alpha_0,\alpha_1,\ldots,\alpha_{n-1},z)$, let $H_z$ be the hyperplane defined by $\alpha(z)$, and let
    \[
    f_{\alpha(z)}(x_2,\ldots,x_n) = f\left(-\sum\limits_{1<j\leq n} \frac{\alpha(z)_j}{\alpha_1} x_j - \frac{\alpha_0}{\alpha_1}, x_2,\ldots,x_n\right)
    \]
    be the restriction of $f$ to $H_z$. Looking at $f_{\alpha(z)}$ as a function of $x_2,\ldots,x_n$ and $z$, we get that the monomial
    \[
    z^{d_1}x_n^{d_1} \prod\limits_{i=2}^{t} x_i^{d_i}
    \]
    appears in $f_{\alpha(z)}$. This is because $f_{\alpha(z)}$ is the same as $f\circ B$ when we replace $x_1$ with $\ell(z) x_n$ for some linear function $\ell(z)$.
    Thus, there are at least $q-d_1$ choices for $z$ to make that monomial survive in $H_z$, in which case we would have that $H_z\in\mathcal{H}$. Since by our assumption
    there is at most $1$ such $z$, we get that $d_1 = q-1$.

    We now observe that the monomial $M = \prod\limits_{i=1}^{t} x_i^{d_i}$ must be in the support of $f$. Indeed, to have the monomial
    $z^{d_1}x_n^{d_1} \prod\limits_{i=2}^{t} x_i^{d_i}$ in $f_{\alpha(z)}$, as $d_1 = q-1$, we must have a monomial whose degree in $x_1$
    is full (i.e. $q-1$), and expanding $\left(\sum\limits_{i=2}^{n}\frac{\alpha(z)_{i}}{\alpha(z)_1}x_i\right)^{q-1}$ (which would be what
    that monomial gives on $x_1$), we must have that the contribution from it would have full degree in $z$, i.e. it must pick the term
    $\alpha(z)_n^{q-1} x_n^{q-1}$. This says that this part of the monomial does not contribute any $x_j$ factors for $j>1$, and hence those
    must be contributed form the original monomial itself.

    Thus, we now have that $d_1 = q-1$, and $M\in{\sf supp}(f)\setminus {\sf supp}(\mathcal{B})$. We now move on to the case $c_{\alpha}=2$, and consider this
    monomial $M$ and whether it stays alive in $f_{\alpha(z)}$. We look at the corresponding hyperplane as
    \[
    -x_2 = \sum\limits_{j>2}\frac{\alpha_j}{\alpha_2}x_j + \frac{\alpha_0}{\alpha_2},
    \]
    and look at $f_{\alpha(z)}$ and in particular in the monomial $z^{d_2}x_n^{d_2} x_1^{q-1} \prod\limits_{i=3}^{t} x_i^{d_i}$. There are a few cases
    that have to be considered.
    \begin{enumerate}
      \item If it exists in $f_{\alpha(z)}$, we get that the coefficient of $x_n^{d_2} x_1^{q-1} \prod\limits_{i=3}^{t} x_i^{d_i}$ is a non-zero polynomial in $z$
    of degree at most $d_2$, and for each $z$ for it is non-zero we get that $H_z\in\mathcal{H}$. Thus there must be a unique choice for $z$ that would make it alive
    and necessarily $d_2 = q-1$. We continue to the next $c$.
      \item Otherwise, it means it has been canceled by some other monomial in $f$. We note that any such monomial must be of the form
      \[
            M' = x_1^{q-1} x_2^{d_2'}\cdots x_{t}^{d_t'},
      \]
      where $d_2'> d_2$.

      We argue that $M'\not\in{\sf supp}(\mathcal{B})$. Indeed, assume towards contradiction this is not the case.
      For this monomial to cancel $M$, we look at what happens when we plug in $x_2$ as in $\alpha(z)$:
      \[
      x_2^{d_2'} =
      \left(\sum\limits_{j>2}\frac{\alpha_j(z)}{\alpha_2}x_j + \frac{\alpha_0}{\alpha_2}\right)^{d_2'}
      =\left(\frac{z}{\alpha_2}x_n + S\right)^{d_2'}
      =\sum\limits_{r\leq d_2'}{d_2'\choose r}\left(\frac{z}{\alpha_2}x_n\right)^{r}S^{d_2'-d_2},
      \]
      where $S = \sum\limits_{2<j<n}\frac{\alpha_j(z)}{\alpha_2}x_j + \frac{\alpha_0}{\alpha_2}$. The contribution from this that may cancel $M$
      comes from $r$, so it is
      \[
      {d_2'\choose d_2}\left(\frac{z}{\alpha_2}x_n\right)^{d_2}S^{d_2'-d_2}.
      \]
      By Lucas's theorem, for ${d_2'\choose d_2}$ to be non-zero mod $p$ (in which case the last expression is $0$ as the characteristic of $\mathbb{F}_q$ is $p$)
      we need $d_2'$ to dominate $d_2$ in the $p$-basis. We then expand $S^{d_2'-d_2}$, and should get from it $\prod\limits_{i=3}^{t} x_i^{d_i-d_i'}$. We will
      do so under the assumption that $d_i\geq d_i'$ for $i\geq 3$; the argument is similar otherwise. For example, if $d_3<d_3'$, then below every occurrence of
      the difference $(d_3-d_3')$ is to be replaced by $(q-1 + d_3 - d_3')$.

      Doing
      the analysis term by term, we should have that $d_2'-d_2$ dominates $d_3-d_3'$ in the $p$-basis, and setting $e_i = (d_2'-d_2) - \sum\limits_{j=3}^{i}(d_j-d_j')$, we should have that $e_i$ dominates $d_{i+1} - d_{i+1}'$ in the $p$-basis. Eventually, we must have that $e_t = 0$.

      We now use the monomial spreading, i.e. Lemma~\ref{lem:spreading}. As $d_2'$ dominated $d_2$, we may get that the monomial
      \[
      M'' = x_1^{q-1}x_n^{d_2}x_2^{d_2'-d_2} x_3^{d_3'} x_{4}^{d_4'}\cdots x_{t}^{d_t'}
      =x_1^{q-1}x_n^{d_2}x_2^{e_2} x_3^{d_3'} x_{4}^{d_4'}\cdots x_{t}^{d_t'}
      \]
      is in ${\sf supp}(\mathcal{B})$.
      As $e_2$ dominates $d_3-d_3'$, we conclude again using Lemma~\ref{lem:spreading}
      that the monomial
      \[
      M'''= x_1^{q-1}x_n^{d_2}x_2^{e_2-(d_3-d_3')}x_3^{d_3} x_{4}^{e_3}\cdots x_{t}^{d_t'}
      = x_1^{q-1}x_n^{d_2}x_3^{d_3} x_2^{e_3} x_{4}^{d_4}\cdots x_{t}^{d_t'},
      \]
      is in ${\sf supp}(\mathcal{B})$. Continuing in this way, we eventually conclude that $M\in {\sf supp}(\mathcal{B})$, and contradiction.

      It follows that we had $M'\not\in{\sf supp}(\mathcal{B})$, and we start the iteration for $c=2$ again with $M'$. Clearly, we will get stuck at $c=2$ at most $q-1$
      as the degree of $x_2$ increases each time, hence eventually we will hit $d_2=q-1$ and proceed to the next variable.
    \end{enumerate}

    Hence, we conclude that under the assumption of the lemma and $x^{\star} = 0$, we have that the monomial $\prod\limits_{i=1}^{t} x_i^{q-1}$ must
    appear in ${\sf supp}(f)$. Define $g = f + s 1_{x=x^{\star}}$ for some $s\in\mathbb{F}_q$ such that $\prod\limits_{i=1}^{t} x_i^{q-1}\not\in {\sf supp}(g)$;
    this is clearly possible, as the support of $1_{x=x^{\star}}$ is full. We would get that the set of hyperplanes $H$ for which $g|_{H}\not\in \mathcal{F}$
    is contained in $\mathcal{H}$, as we only changed $f$ in $x^{\star}$ and any $H\not\in\mathcal{H}$ does not contain it. Hence, $\eps_k(g)\leq \eps_k(f) = 1/q$.
    We claim that $g\in\mathcal{F}$. Indeed, otherwise we would run the above argument on $g$ and conclude that $g$ must contain the monomial $\prod\limits_{i=1}^{t} x_i^{q-1}$
    in its support, which is clearly impossible. This is a contradiction, and therefore $g\in\mathcal{F}$ as desired.
\end{proof}

The following corollary is immediate:
\begin{corollary}\label{lem:relate2}
  Let $p$ be prime, $q\in\mathbb{N}$ be a power of $p$ and let $t\in\mathbb{N}$.
  Let $\mathcal{B}\subseteq \set{g\colon\mathbb{F}_q^t\to\mathbb{F}_q}$ be an affine invariant code,
  and denote $\mathcal{F} = {\sf Lift}_{k+1}(\mathcal{B})$. Suppose that $k\geq t$, and let $f\colon\mathbb{F}_q^{k+1}\to\mathbb{F}_q$
  be such that $f\not\in\mathcal{F}$. Then if $k\geq k'\geq t$, then $\eps_{k}(f)\leq q^{k-k'}(f)$.
\end{corollary}
\subsection{Proof of Theorem~\ref{thm:main2}}
In this section, we explain how to adapt the argument in Section~\ref{sec:pf_of_main} to prove Theorem~\ref{thm:main2}. First, the set $S$ in this context
is defined to be
\[
S = \sett{A}{{\sf dim}(A) = t, f|_{A}\not\in\mathcal{B}}.
\]
For the argument there we need $t$ to be a sufficiently large constant, say larger than $M$, and we claim we may indeed assume that. Indeed, otherwise we may look
at the $t+M$ flat tester and get that $\mu(S\uparrow^{t+M})\leq q^M\mu(S)$ is the rejection probability (where we used Corollary~\ref{lem:relate2}). We then look
at the problem as trying to understand the lifted code of $\mathcal{B} = {\sf Lift}_{t+M}(\mathcal{B})$, which is also affine invariant, and we now have that the new
$t$ is large enough. We henceforth assume that $t$ is large enough to begin with.

Claims~\ref{claim:reduce_to_expansion},~\ref{claim:pseudo_zoom_out} remain unchanged except that we appeal to Lemma~\ref{lem:HSS2} instead
of Lemma~\ref{lem:relate}. The proof of Claim~\ref{claim:pseduo_zoom_in_linear} also remains unchanged, except that in the end we appeal again to Lemma~\ref{lem:HSS2} instead of Lemma~\ref{lem:relate}. This establishes Proposition~\ref{prop:find_zoom_in}
in this case.

The discussion before Claim~\ref{claim:is_empty} and the claim itself continue to hold as is in this case, and we explain the slight adaptation to the rest of
the argument in Section~\ref{sec:correct_val}.
\begin{proposition}\label{prop:fix_x_star_2}
  There exists $c\in\mathbb{F}_q$ such that changing the value of $f(x^{\star})$ to $c$, we have that
  \[
  \cProb{A''\text{ $t$ flat}}{x^{\star}\in A''}{{\sf deg}(f|_{A''})}\geq \frac{1}{2q}.
  \]
\end{proposition}
\begin{proof}
  Take any $(t+1)$-flat $A'\subseteq A$ containing $x^{\star}$, and define $g = f|_{A'}$. We claim that we may change $f$ at $x^{\star}$ and have
  that $g\in\mathcal{F}$. Otherwise, from the second item in Lemma~\ref{lem:HSS2}, the fraction of $t$-flats $B\subseteq A'$ such that $g'|_{B}\not\in\mathcal{B}$ is
  larger than $1/q$. As the fraction of $B$'s that contain $x^{\star}$ is exactly $1/q$, it follows that there is $B\subseteq A'$ not containing $x^{\star}$ such that
  $g|_{B} \not\in\mathcal{B}$. But for such $B$'s we have $g|_{B} = f|_{B}$, and contradiction.

  Sampling $A$ a $(t+100)$ flat containing $x^{\star}$ randomly and then a $t$-flat $A''\subseteq A$ containing $x^{\star}$, we get that with probability at least $1/2$ we may change $f(x^{\star})$ and have $f|_{A''}\in\mathcal{B}$. Thus, taking the majority vote we may choose $f(x^{\star})$ that appeases at least $\frac{1}{2q}$ of the $t$-flats
  containing $x^{\star}$.
\end{proof}

Given Proposition~\ref{prop:fix_x_star_2}, Section~\ref{sec:iterate} goes through as well, completing the proof of Theorem~\ref{thm:main2}.
\qed

\section{Discussion and open questions}
Our work explores a potential connection between testing questions in codes and expansion in the underlying test graph,
using the idea that the error set exhibits some non sharp-threshold type behaviour. This connection highlights several
problems that we think may be of interest.
\begin{enumerate}
  \item Stability results for Kruskal-Katona type theorems. What can we say about the structure of small sets $S\subseteq V_q(k,\ell)$ for which $\mu(S\uparrow)\leq q\mu(S)$?
  Using our techniques, it follows that such sets must be correlated with a zoom-in set or a zoom-in with respect to the linear part (which we are able to eliminate in our case),
  but it would be interesting to get a more thorough understanding of this problem. Similarly, it would be interesting to understand the structure of large sets with non-perfect shadow, i.e. $\mu(S\uparrow)\leq 1-\delta$.
  \item Beyond lifted codes. Can we use expansion type results on structures such as the Grassmann graph (but maybe more) to prove more testing results on other codes?
  As we have seen, the proof goes through relatively easily for the class of lifted affine invariant codes (improving the dependency on the field size $q$ over the result of~\cite{HRZS}), and we suspect our method should apply in other settings as well.
  \item Characterization of near degree $1$ functions on the Affine Grassmann graph. As shown in Lemma~\ref{lem:high_wt_lvl1}, small sets $S$ for which $1-\Phi(S)\geq 1/q$
  have almost all of their Fourier degree on the first level. In this case, we establish a relatively weak structural result, and it is tempting to ask whether a more detailed
  structural result holds in this case similarly to the classical FKN theorem from the Boolean cube~\cite{FKN}.
  \item Beyond the $99\%$ regime. Can the approach suggested herein, or similar ones, be applied to study the testing question for the Reed-Muller code wherein the success
  probability of the tester is only guaranteed to be at least $1/q + \delta$, i.e. the notorious $1\%$ regime?
\end{enumerate}

\bibliographystyle{plain}
\bibliography{ref}
\appendix
\section*{Appendix}
This section is devoted to the proof of Theorem~\ref{thm:expansion}. Our approach closely follows the approach in~\cite{KMS2},
however as we are only concerned with the special case associated with zoom-ins/ zoom-outs of dimension/ co-dimension $1$, our analysis
is considerably simpler. Roughly speaking, our prof consists of the following three components:
\begin{enumerate}
  \item First, we define a Cayley graph that closely resembles the affine Grassmann graph, and show that studying expansion over the two is roughly
  equivalent (up to some loss in the parameters).
  \item Second, we show that for the expansion parameters in question, the problem reduces to studying the structure of functions that have almost all
  of their Fourier mass on the first level component in the natural degree decomposition.
  \item Finally, we perform a $4$th-moment vs $2$nd-moment type analysis and deduce the structural result.
\end{enumerate}
Throughout this section, we think of $W$ as a linear space over $\mathbb{F}_q$ with dimension $k$; without loss of generality
$W = \mathbb{F}_q^k$. We consider the affine Grassmann graph over $\ell$-flats.
\section{The Cayley graph construction}
Consider the edge-weighted graph $H = (V,E)$ defined as follows. The set of vertices $V$ consist of tuples
$(s,x_1,\ldots,x_{\ell})$ where $s,x_1,\ldots,x_{\ell}\in \mathbb{F}_q^k$.
The edge weights are described according to the following randomized process; to sample a neighbour of $(s,x_1,\ldots,x_{\ell})$:
\begin{enumerate}
  \item sample $y\in \mathbb{F}_q^k$ uniformly;
  \item sample $b_0,b_1,\ldots,b_{\ell}\in\mathbb{F}_q$ uniformly;
  \item output $(s+b_0y,x_1+b_1 y\ldots, x_{\ell} + b_{\ell} y)$.
\end{enumerate}
Given a set of vertices in the affine Grassmann graph $S\subseteq V(\mathbb{F}_q^k, \ell)$, we associate with it the set $S^{\star}$ in the Cayley graph
defined as
\[
S^{\star} = \sett{(s,x_1,\ldots,x_{\ell})}{s+{\sf span}(x_1,\ldots,x_{\ell})\in S}.
\]

We establish some properties of $S$ and $S^{\star}$. First, we show that the non-expansion of $S^{\star}$ may be lower bounded
by the non-expansion of $S$ (in fact the two are close, but we only need this direction).
\begin{claim}\label{claim:exp_lb}
  $1-\Phi(S^{\star})\geq 1-\Phi(S) - q^{-\ell}$.
\end{claim}
\begin{proof}
  Recall that $1-\Phi(S^{\star})$ is the probability that starting from a random vertex in $S^{\star}$ and taking a step, we stay in the set $S^{\star}$.
  Denote by $v = (s,x_1,\ldots,x_{\ell})$ the starting point of the walk, by $y,b_0,\ldots,b_{\ell}$ the parameters that define the step of the walk, and
  by $u$ the endpoint of the random walk. There are a few cases:
  \begin{enumerate}
    \item $b_0 = b_1=\ldots=b_{\ell} = 0$, which happens with probability $q^{-(\ell+1)}$ and corresponds to a self-loop.
    \item $b_0\neq 0$, $b_1=\ldots = b_{\ell} = 0$, which corresponds to the case the hyperplanes defined by $v,u$ are parallel. This happens with
    probability $\frac{q-1}{q^{\ell+1}}$.
    \item ${\sf span}(x_1 + b_1y,\ldots,x_{\ell}+b_{\ell}y)$ has dimension less than $\ell$, which happens with probability at most $q^{\ell-k}$.
    \item Otherwise, $u$ is a random affine space of dimension $\ell$ that intersects $v$ in size $q^{\ell-1}$. This happens with probability
    $(1-q^{-\ell} - q^{\ell-k})$.
  \end{enumerate}
  We note that in the case of the 3rd item, we always escape the set and hence this doesn't contribute to $1-\Phi(S^{\star})$. We compare
  the rest of these probabilities to the corresponding walk on the affine Grassmann graph. Starting at an affine space $V$ of dimension $\ell$,
  going to $K\supseteq U$ of dimension $\ell+1$ and then to a random $U\subseteq K$ of dimension $\ell$, we have:
  \begin{enumerate}
    \item The probability that $U = V$ is $\frac{1}{q^{\ell+1} - 1}\frac{q-1}{q}$.
    \item The probability that $V$ and $U$ are parallel is $\frac{q-1}{\frac{q^{\ell+1} - 1}{q-1} q} = \frac{(q-1)^2}{q(q^{\ell+1} - 1)}$.
    \item Otherwise, $U$ is random affine space of dimension $\ell$ that intersects $V$ in size $q^{\ell-1}$. The probability here is
    $1-\frac{q-1}{q^{\ell+1} -1}$.
  \end{enumerate}
  Looking at the ratios between the probability of a case in the Cayley graph and the probability of a case in the affine Grassmann grah,
  the first two are at least $1$, whereas the last one is at least $1-q^{-\ell}$. Thus,
  \[
  1-\Phi(S^{\star})\geq (1-q^{-\ell})(1-\Phi(S))\geq 1-\Phi(S) - q^{-\ell}.\qedhere
  \]
\end{proof}

Next, we consider the analogous notions of zoom-ins for sets in the Cayley graph.
\begin{definition}
Let $T$ be a set in the Cayley graph.
\begin{enumerate}
  \item For $z\in\mathbb{F}_p^k$, the zoom-in of $T$ with respect to $z$ is the set
  \[
  \sett{(s,x_1,\ldots,x_{\ell})}{z\in s+{\sf span}(x_1,\ldots,x_{\ell})}.
  \]
  \item For $z\in\mathbb{F}_p^k\setminus\set{0}$, the zoom-in of $T$ with respect to $z$ on the linear part is the set
  \[
  \sett{(s,x_1,\ldots,x_{\ell})}{z\in {\sf span}(x_1,\ldots,x_{\ell})}.
  \]
  \item For an affine hyperplane $W\subseteq\mathbb{F}_p^k$, the zoom-out of $T$ with respect to $W$ is the set
  \[
  \sett{(s,x_1,\ldots,x_{\ell})}{ s + {\sf span}(x_1,\ldots,x_{\ell})\subseteq W}.
  \]
  \item
  For a hyperplane $W\subseteq\mathbb{F}_p^k$, the zoom-in of $T$ with respect to $W$ on the linear part is the set
  \[
  \sett{(s,x_1,\ldots,x_{\ell})}{ {\sf span}(x_1,\ldots,x_{\ell})\subseteq W}.
  \]
\end{enumerate}
For each one of these cases, say for zoom-ins, we say that $T$ is $\xi$-pseudo-random with respect to it if
\[
\mu(T_{z})\defeq\frac{\card{\sett{(s,x_1,\ldots,x_{\ell})\in T}{z\in s+{\sf span}(x_1,\ldots,x_{\ell})}}}{\card{\sett{(s,x_1,\ldots,x_{\ell})}{z\in s+{\sf span}(x_1,\ldots,x_{\ell})}}}\leq \xi.
\]
\end{definition}

We now show that notions of pseudo-randomness of $S$ transfer to the same notions of pseudo-randomness for $S^{\star}$.
\begin{claim}
  If $S$ is $\xi$-pseudo-random against zoom-ins, then $S^{\star}$ is $\xi$ pseudo-random with respect to zoom-ins. Same goes for zoom-outs etc.
\end{claim}
\begin{proof}
  Sampling $v = (s,x_1,\ldots,x_{\ell})$ from the Cayley graph conditioned on it representing an affine subspace of dimension $\ell$ and containing $z$,
  the subspace it represents is distributed uniformly among all subspaces containing $z$, hence in $S_z$ with probability $\mu(S_z)\leq \xi$. If $v$ does not
  represent an affine subspace of dimension $\ell$, we clearly have $v\not\in (S^{\star})_z$. Thus,
  \[
  \mu((S^{\star})_z) = \cProb{v}{z\in v}{v\in S^{\star}\land v\text{ is dimension $\ell$}}\leq \cProb{v}{z\in v, v\text{ is dimension $\ell$}}{v\in S^{\star}} = \mu(S_z)\leq \xi.
  \qedhere
  \]
\end{proof}

As a special case of the previous claim, we get that a good zoom-in for $S^{\star}$ (i.e., one on which the measure of this set is almost $1$) is also be good for $S$.
\begin{corollary}
  Suppose that $\mu((S^{\star})_z)\geq 1-\delta$. Then $\mu(S_z)\geq 1-\delta$.
\end{corollary}

\section{Decompositions}
\subsection{The Fourier decomposition}
Let $F = 1_{S^{\star}}$. We shall now think of $F\colon \mathbb{F}_{q}^k\to\{0,1\}$ as a function, and develop it according to the basis of characters.
In this context, writing $q = p^r$ where $p$ is prime, we consider the trace map ${\sf Tr}\colon \mathbb{F}_q\to\mathbb{F}_p$ defined as
${\sf Tr}(a) = \sum\limits_{i=1}^{r-1} a^{p^i}$. A character of $\mathbb{F}_q$ is then defined as $\chi_a(x) = \omega^{{\sf Tr}(ax)}$
for $a\in\mathbb{F}_q$, where $\omega$ is the $p$th root of unity. A character of $\mathbb{F}_q^k$ is indexed by $\vec{a}\in \mathbb{F}_q^k$
and is defined as $\chi_{\vec{a}}(x) = \prod\limits \chi_{a_i}(x_i) = \omega^{\sum\limits{i=1}^{k}{\sf Tr}(a_ix_i)}$.
Finally, a character of $((\mathbb{F}_q)^k)^{\ell+1}$ is indexed by $\alpha = (\alpha_0,\ldots,\alpha_{\ell})\in(\mathbb{F}_q^{k})^{\ell+1}$
and is defined as
\[
\chi_{\alpha}(s,x_1,\ldots,x_{\ell}) = \chi_{\alpha_0}(s) \prod\limits_{i=1}^{\ell}\chi_{\alpha_i}(x_i).
\]
We will use the abbreviation $x = (x_1,\ldots,x_{\ell})$, and then write
\[
F(s,x) = \sum\limits_{\alpha} \widehat{F}(\alpha) \chi_{\alpha}(s,x), \qquad\text{where }
\widehat{F}(\alpha) = \Expect{(s,x)}{F(s,x)\overline{\chi_\alpha(s,x)}}.
\]
\begin{claim}\label{claim:sym}
  Suppose we have $\alpha,\beta$ such that $\alpha_0 = \beta_0$ and
  ${\sf span}(\alpha_0,\alpha_1,\ldots,\alpha_{\ell}) = {\sf span}(\beta_0,\beta_1,\ldots,\beta_{\ell})$.
  Then $\widehat{F}(\alpha) = \widehat{F}(\beta)$.
\end{claim}
\begin{proof}
  Follows as $f$ is invariant under $(s,x_1,\ldots,x_{\ell})\rightarrow (s+z_0, z_1,\ldots,z_{\ell})$
  where $z_1,\ldots,z_{\ell}$ are linearly independent linear combinations of $x_1,\ldots,x_{\ell}$, and $z_0$ is a linear combination of $x_1,\ldots,x_{\ell}$.
\end{proof}

Next, we calculate the eigenvalues of the characters with respect to the random walk on the Cayley graph.
\begin{claim}
  Let $\alpha = (\alpha_0,\alpha_1,\ldots,\alpha_{\ell})$ be such that ${\sf dim}({\sf span}(\alpha_0,\ldots,\alpha_{\ell})) = d$. Then
  $\chi_{\alpha}$ is an eigenfunction with respect to the random walk on the Cayley graph with eigenvalue $q^{-d}$.
\end{claim}
\begin{proof}
  The eigenvalue is easily seen to be equal to $\Expect{b_0,b_1,\ldots,b_{\ell}, y}{\chi_{\sum\limits_{i=0}^{\ell} b_i\alpha_i}(y)}$.
  Note that if the dimension of ${\sf span}(\alpha_0,\ldots\alpha_{\ell})$ is $d$, then the probability that $\sum\limits_{i=0}^{\ell} b_i\alpha_i = 0$
  is $q^{-d}$. In that case, the expectation is $1$, and otherwise it is $0$.
\end{proof}

\subsection{The level decomposition}
For $i=0,1,\ldots, \ell$, define
\[
F_{{\sf lin}, i} (s,x) =
\sum\limits_{\substack{\alpha: \alpha_0\in{\sf span}(\alpha_1,\ldots,\alpha_{\ell})  \\ {\sf dim}({\sf span}(\alpha_1,\ldots,\alpha_{\ell})) = i}} \widehat{F}(\alpha)\chi_{\alpha}(s,x),
\qquad
F_{{\sf aff}, i} (s,x) =
\sum\limits_{\substack{\alpha: \alpha_0\not\in  {\sf span}(\alpha_1,\ldots,\alpha_{\ell})\\ {\sf dim}({\sf span}(\alpha_1,\ldots,\alpha_{\ell})) = i}} \widehat{F}(\alpha)\chi_{\alpha}(s,x),
\]
and for simplicity $F_i(s,x) = F_{{\sf lin}, i} (s,x) + F_{{\sf aff}, i-1} (s,x)$. Clearly
\[
F(s,x) = \sum\limits_{i=0}^{\ell} F_i(s,x).
\]

Denoting by $H$ the normalized adjacency operator of the Cayley graph, we have that
$H F(s,x) = \sum\limits_{i=0}^{\ell} q^{-i} F_i$, and so
\[
1-\Phi(S^{\star}) = \frac{1}{\mu(S^{\star})}\inner{F}{H F} =
\frac{1}{\mu(S^{\star})}
\left(\norm{F_{0}}_2^2 + \sum\limits_{i=0}^{\ell} q^{-i}\norm{F_{i}}_2^2 + q^{-\ell-1} \norm{F_{{\sf aff},\ell}}_2^2\right).
\]
As $1-\Phi(S^{\star})\geq \frac{1}{q} - \frac{1}{q^{\ell}}$ from Claim~\ref{claim:exp_lb}, $\norm{F_{0}}_2 = \mu(S^{\star})$
and $\sum\limits_{i\geq 2}\norm{F_{i}}_2^2  = \mu(S^{\star}) - \mu(S^{\star})^2 - \norm{F_1}_2^2$
by Parseval, we get that
\[
\frac{1}{q} - \frac{1}{q^{\ell}}
\leq \frac{1}{\mu(S^{\star})}\left(\mu(S^{\star})^2 + \frac{1}{q}\norm{F_{1}}_2^2 + \frac{1}{q^2}(\mu(S^{\star}) - \mu(S^{\star})^2 - \norm{F_{1}}_2^2)\right).
\]
Rearranging we get
\[
\frac{1}{q} - \frac{1}{q^2} - \frac{1}{q^{\ell}}
\leq \frac{1}{\mu(S^{\star})}\left(\frac{1}{q} - \frac{1}{q^2}\right)\norm{F_{1}}_2^2
+\mu(S^{\star}),
\]
and so
\[
\frac{\norm{F_{1}}_2^2}{\mu(S^{\star})}
\geq 1-q^{2-\ell} - q^2\xi.
\]
We summarize this discussion with the following lemma.
\begin{lemma}\label{lem:high_wt_lvl1}
  Let $S$ be as in Theorem~\ref{thm:expansion}, and let $S^{\star}$ be the corresponding set in the Cayley graph.
  Then letting $F = 1_{S^{\star}}$ and looking at the level decomposition  above, we have
  \[
  \frac{\norm{F_{1}}_2^2}{\mu(S^{\star})}\geq 1-q^{2-\ell} - q^2\xi.
  \]
\end{lemma}

\subsection{Lower bounding the fourth norm of $F_1$ and stating the upper bound}
We now move on to the heart of the argument which handles the fourth norm of $F_1$. First,
we show an easy lower bound on it:
\begin{corollary}\label{corr:4th_norm_lb}
  $\frac{\norm{F_1}^4_4}{\mu(S^{\star})}\geq \left(1-q^{2-\ell} - q^2\xi\right)^4$.
\end{corollary}
\begin{proof}
  By H\"{o}lder's inequality we have
  \[
  \norm{F_1}_2^2 = \inner{F_1}{F_1} = \inner{F_1}{F} \leq \norm{F_1}_4\norm{F}_{4/3} = \norm{F_1}_4\mu(S^{\star})^{3/4},
  \]
  using the lower bound on the left hand side from Lemma~\ref{lem:high_wt_lvl1} establishes the claim.
\end{proof}

Next, we state the upper bound on it, and then show how the two bounds imply Theorem~\ref{thm:expansion}. The rest of the appendix is then devoted into proving
this upper bound.
\begin{lemma}\label{lem:4th_norm_ub}
  Suppose $S$ is
  \begin{enumerate}
    \item $\xi$ pseudo-random with respect to zoom-outs (as well as on its linear part),
    \item $\mu(S)\leq \xi$,
    \item $\xi$ pseudo-random zoom ins with respect to their linear part,
    \item $a$ pseudo-random with respect to zoom-ins.
  \end{enumerate}
  Then
  \[
  \norm{F_1}_4^4\leq \mu(S^{\star})a^2 + 863 q^{2}\mu(S^{\star})\xi^{1/4}.
  \]
\end{lemma}
We now show the quick derivation of Theorem~\ref{thm:expansion}.
\begin{proof}[Proof of Theorem~\ref{thm:expansion}]
  Combining Corollary~\ref{corr:4th_norm_lb} and Lemma~\ref{lem:4th_norm_ub} we get that
  \[
  a^2 + 863 q^{2} \xi^{1/4}\geq 1-4q^{2-\ell}-4q^2\xi,
  \]
  so $a\geq 1-q^{2}(867\xi^{1/4} + q^{-\ell})$ provided $\xi$ is small enough with respect to $q$ ($\xi\leq q^{-10}$ will do).
\end{proof}

\subsection{An alternative description to $F_1$}
To handle $F_1$, we shall need a different combinatorial description for $F_1$. Define
$f_{1,{\sf lin}}, f_{1,{\sf aff}}\colon\mathbb{F}_q^{\ell}\to [-1,1]$ as
\[
f_{1,{\sf lin}}(x) = \mu((S^{\star})_{{x,{\sf lin}}}) - \mu(S^{\star}),
\qquad
f_{1,{\sf aff}}(x) = \mu((S^{\star})_{{x,{\sf aff}}}) - \mu(S^{\star}).
\]

Let $\mathcal{M} = \mathbb{F}_{q}^{\ell}\setminus\set{0}$. We define the equivalence relation on $\mathcal{M}$ which is
$M\sim M'$ if $M = i M'$ for some $i\in\mathbb{F}_q$, and let $\mathcal{B}$ be the equivalency classes of this relations;
we choose a representative element from each equivalency class (arbitrarily).
\begin{claim}
  $F_1(s,x) = \sum\limits_{M\in \mathcal{B}} f_{1,{\sf lin}}(\inner{M}{x})
  +
  \sum\limits_{M\in \mathbb{F}_q^{\ell}} f_{1,{\sf aff}}(s + \inner{M}{x})$.
\end{claim}
\begin{proof}
  By definition, we have
  \begin{align*}
  F_{1,{\sf lin}}(s,x)
  &=
  \sum\limits_{\beta\in \mathcal{B}}
  \sum\limits_{M\in\mathcal{M}, v\in \mathbb{F}_q}
  \widehat{F}(v\beta, M_1\beta,\ldots,M_{\ell}\beta) \chi_{v\beta,M_1\beta,\ldots,M_{\ell}\beta}(s,x).
  \end{align*}
  We split this sum according to $v=0$ and $v\neq 0$.
  \paragraph{Contribution from $v=0$.}
  For $v=0$, using Claim~\ref{claim:sym} we get contribution of
  \begin{align*}
  &\sum\limits_{\beta\in\mathcal{B}}
  \sum\limits_{M\in\mathcal{M}}
  \widehat{F}(0, \beta,0,\ldots,0) \chi_{0,M_1\beta,\ldots,M_{\ell}\beta}(s,x)\\
  &= \frac{1}{q-1}\sum\limits_{\beta\in\mathbb{F}_q^{k}\setminus\set{0}}
  \Expect{s',x'}{F(s',x')\chi_{\beta}(-x_1')}
  \sum\limits_{M\in\mathcal{M}}\chi_{\beta}\left(\sum\limits_{i=1}^{\ell} M_i x_i\right)\\
  &=
  \frac{1}{q-1}
  \Expect{s',x'}{F(s',x')
  \sum\limits_{M\in\mathcal{M}}
  \sum\limits_{\beta\in\mathbb{F}_q^{k}\setminus\set{0}}\chi_{\beta}\left(-x'_1 + \sum\limits_{i=1}^{\ell} M_i x_i\right)}.
  \end{align*}
  Adding $\frac{\card{\mathcal{M}}}{q-1} \mu(S^{\star})$ to this expression amounts to also including $\beta = 0$, hence we get that
  \[
  \frac{\card{\mathcal{M}}}{q-1} F_0 + F_{1,{\sf lin}}(s,x)
  = \frac{1}{q-1}
  \Expect{s',x'}{F(s',x')
  \sum\limits_{M\in\mathcal{M}}
  \sum\limits_{\beta\in\mathbb{F}_q^{k}}\chi_{\beta}\left(-x'_1 + \sum\limits_{i=1}^{\ell} M_i x_i\right)}.
  \]

  If $-x'_1 + \sum\limits_{i=1}^{\ell} M_i x_i\neq 0$, the sum over $\beta$ is $0$ and otherwise it is $q^k$, so we get
  \begin{align*}
  \frac{q^{k}}{q-1}
  \Expect{s',x'}{F(s',x')
  \sum\limits_{M\in\mathcal{M}}
  1_{x'_1 = \sum\limits_{i=1}^{\ell} M_i x_i}
  }
  &=
  \frac{q^{k}}{{q-1}} \sum\limits_{M\in\mathcal{M}}
  \Expect{s',x'}{F(s',x') 1_{x_1' = \inner{M}{x}}}\\
  &=\frac{1}{{q-1}}\sum\limits_{M\in\mathcal{M}} \mu((S^{\star})_{\inner{M}{x},{\sf lin}})\\
  &=\sum\limits_{M\in\mathcal{B}} \mu((S^{\star})_{\inner{M}{x},{\sf lin}}).
  \end{align*}

  \paragraph{Contribution from $v\neq 0$.} For $v\neq 0$ we get from similar computations that the contribution is
  \[
  \sum\limits_{\beta\in\mathcal{B}}
  \sum\limits_{M\in\mathcal{M}, v\neq 0}
  \widehat{F}(v\beta, 0,\ldots,0) \chi_{v\beta,M_1\beta,\ldots,M_{\ell}\beta}(s,x)
  =\sum\limits_{\beta\in\mathbb{F}_q^k\setminus\{0\}}
  \sum\limits_{M\in\mathcal{M}}
  \widehat{F}(\beta, 0,\ldots,0) \chi_{\beta,M_1\beta,\ldots,M_{\ell}\beta}(s,x).
  \]
  We add to that $F_{1,{\sf aff}}$, which is the term corresponding to $M = 0$; we then add $q^{\ell} F_0$, which corresponds to taking $\beta = 0$ as well.
  Hence we get that the contribution from $v\neq 0$ plus $F_{1,{\sf aff}}(s,x) + q^{\ell}\mu(S^{\star})$ is
  \begin{align*}
  &\sum\limits_{\beta\in\mathbb{F}_q^{k}}
  \Expect{s',x'}{F(s',x')\chi_{\beta}(-s')}
  \sum\limits_{M\in\mathbb{F}_q^{\ell}}\chi_{\beta}\left(s + \sum\limits_{i=1}^{\ell} M_i x_i\right)\\
  &=
  \Expect{s',x'}{F(s',x')
  \sum\limits_{M\in \mathbb{F}_q^{\ell}}
  \sum\limits_{\beta\in\mathbb{F}_q^{k}}\chi_{\beta}\left(s + \sum\limits_{i=1}^{\ell} M_i x_i -s'\right)}.
  \end{align*}
  If $-s' + s +\sum\limits_{i=1}^{\ell} M_i x_i = 0$ we get that the sum over $\beta$ is $q^{k}$ and otherwise it is
  $0$. Hence we get
  \begin{align*}
  q^{k}\Expect{s',x'}{f(s',x')
  \sum\limits_{M\in \mathbb{F}_q^{\ell}} 1_{-s' +s+ \sum\limits_{i=1}^{\ell} M_i x_i = 0}}
  &=q^{k} \sum\limits_{M\in\mathbb{F}_q^{\ell}}\Expect{s',x'}{f(s',x')1_{\inner{M}{x} + s = s'}}\\
  &=\sum\limits_{M\in\mathbb{F}_q^{\ell}} \mu((S^{\star})_{s + \inner{M}{x}}).
  \end{align*}
  Combining all, and moving the multiples of $\mu(S^{\star})$ we have added to the other side, we get that
  \[
  F_{1,{\sf lin}}(s,x) + F_{1,{\sf aff}}(s,x) =
  \sum\limits_{M\in\mathcal{B}} f_{1,{\sf lin}}(\inner{M}{x})
  +
  \sum\limits_{M\in\mathbb{F}_q^{\ell}} f_{1,{\sf aff}}(s + \inner{M}{x}).\qedhere
  \]
\end{proof}
\section{Properties of $f_{1, {\sf lin}}$ and $f_{1, {\sf aff}}$}
\subsection{Orthogonality and symmetries}
\begin{claim}\label{claim:f_orth}
  We have
  \[
  \Expect{x\in\mathbb{F}_q^k\setminus{0}}{f_{1,{\sf lin}}(x)} = 0,
  \qquad\qquad
  \Expect{x\in\mathbb{F}_q^k}{f_{1,{\sf aff}}(x)} = 0.
  \]
\end{claim}
\begin{proof}
  This is obvious by the definition of these functions.
\end{proof}

\subsection{Second moment}
\begin{claim}\label{claim:f_2nd}
We have
  \[
  \Expect{x\in\mathbb{F}_q^k\setminus{0}}{f_{1,{\sf lin}}(x)^2} \leq \frac{\norm{F_1}_2^2}{\card{\mathcal{B}}}\leq \frac{\mu(S^{\star})}{\card{\mathcal{B}}},
  \qquad\qquad
  \Expect{x\in\mathbb{F}_q^k}{f_{1,{\sf aff}}(x)^2} \leq \frac{\norm{F_1}_2^2}{q^{\ell}}\leq \frac{\mu(S^{\star})}{q^{\ell}}.
  \]
\end{claim}
\begin{proof}
  Expanding $\norm{F_1}_2^2$, it is equal to
  \[
  \Expect{s,x}{\sum\limits_{M\in\mathcal{B}}{\card{f_{1,{\sf lin}}(\inner{M}{x})}^2}
  +\sum\limits_{M\in\mathcal{B}, M'\in\mathbb{F}_q^{\ell}}f_{1,{\sf lin}}(\inner{M}{x})f_{1,{\sf aff}}(s+\inner{M}{x})
  +\sum\limits_{M'\in\mathbb{F}_q^{\ell}}\card{f_{1,{\sf aff}}(s+\inner{M}{x})}^2}.
  \]
  We note that for each $x$, $M,M'$, the expectation of $f_{1,{\sf lin}}(\inner{M}{x})f_{1,{\sf aff}}(s+\inner{M}{x})$ over $z$ is $0$ by Claim~\ref{claim:f_orth},
  hence the middle sum vanishes. The other two sums are non-negative so it follows that each one of them is at most $\norm{F_1}_2^2$ in expectation, and the claim
  follows by translating them into expectations
\end{proof}
\subsection{Fourier coefficients}
We shall now think of $f_{1,{\sf lin}}, f_{1,{\sf aff}}$ as functions from $\mathbb{F}_q^k$ to $\mathbb{R}$ and may therefore discuss their Fourier coefficients.
\begin{claim}\label{claim:f_Fourier}
  Let $\alpha\in \mathbb{F}_q^k$ index a Fourier coefficient. Then
  \[
  \widehat{f_{1,{\sf lin}}}(\alpha) = \frac{1}{q-1}\widehat{F}(0,\alpha,\ldots,\alpha),
  \qquad\qquad
  \widehat{f_{1,{\sf aff}}}(\alpha) = \widehat{F}(\alpha,0,\ldots,0).
  \]
\end{claim}
\begin{proof}
  By definition
  \begin{equation}\label{eq:equate_coeff}
  F_1(s,x)
  =\sum\limits_{M\in\mathcal{B}} f_{1,{\sf lin}}(\inner{M}{x})
  +\sum\limits_{M\in\mathbb{F}_q^{\ell}} f_{1,{\sf aff}}(s+\inner{M}{x}),
  \end{equation}
  and we expand the right hand side, as well as the left hand side, according to Fourier decomposition. The first term on the right hand side is equal to
  \[
  \frac{1}{q-1}\sum\limits_{M\in\mathcal{M}} \sum\limits_{\alpha\in\mathbb{F}_q^k} \widehat{f_{1,{\sf lin}}}(\alpha) \chi_{\alpha}(\inner{M}{x}).
  \]
  We have
  \[
  \chi_{\alpha}(\inner{M}{x})
  = \chi_{\alpha}(M_1x_1+\ldots+M_{\ell}x_{\ell})
  = \chi_{M_1\alpha}(x_1)\cdots\chi_{M_{\ell}\alpha}(x_{\ell})
  = \chi_{(0,M\alpha )}(s,x),
  \]
  hence
  \begin{equation}\label{eq_1}
  \sum\limits_{M\in\mathcal{B}} f_{1,{\sf lin}}(\inner{M}{x})=
  \frac{1}{q-1}\sum\limits_{M\in\mathcal{M}} f_{1,{\sf lin}}(\inner{M}{x})
  =\frac{1}{q-1}\sum\limits_{M\in\mathcal{M}} \sum\limits_{\alpha\in\mathbb{F}_q^k} \widehat{f_{1,{\sf lin}}}(\alpha)  \chi_{(0,\alpha M)}(s,x).
  \end{equation}
  Similarly, we get that
  \begin{equation}\label{eq_2}
  \sum\limits_{M\in\mathbb{F}_q^{\ell}} f_{1,{\sf aff}}(s+\inner{M}{x})
  =
  \sum\limits_{M\in\mathbb{F}_q^{\ell}} \sum\limits_{\alpha\in\mathbb{F}_q^k} \widehat{f_{1,{\sf aff}}}(\alpha) \chi_{\alpha}(s+\inner{M}{x})
  =
  \sum\limits_{M\in\mathbb{F}_q^{\ell}} \sum\limits_{\alpha\in\mathbb{F}_q^k} \widehat{f_{1,{\sf aff}}}(\alpha) \chi_{(\alpha,M\alpha)}(s,x).
  \end{equation}

  Finally, we have by definition that
  \[
  F_1(s,x)
  =\sum\limits_{\substack{\alpha: \alpha_0\in{\sf span}(\alpha_1,\ldots,\alpha_{\ell})  \\ {\sf dim}({\sf span}(\alpha_1,\ldots,\alpha_{\ell})) = 1}} \widehat{F}(\alpha)\chi_{\alpha}(s,x)+
  \sum\limits_{\alpha_0\in\mathbb{F}_q^k\setminus\set{0}} \widehat{F}(\alpha_0,0,\ldots,0)\chi_{\alpha_0,0,\ldots,0}(s,x).
  \]
  Expanding the first sum and using Claim~\ref{claim:sym} we get it is equal to
  \[
  \sum\limits_{\alpha_0\in\mathbb{F}_q^k\setminus\set{0}, M\in\mathcal{M}} \widehat{F}(0, M\alpha_0)\chi_{0, M\alpha_0}(s,x)
  +
  \sum\limits_{\alpha_0\in\mathbb{F}_q^k\setminus\set{0}, M\in\mathcal{M}} \widehat{F}(\alpha_0, M\alpha_0)\chi_{\alpha_0, M\alpha_0}(s,x).
  \]
  Hence
  \begin{equation}\label{eq_3}
  F_1(s,x) =
  \sum\limits_{\substack{\alpha_0\in\mathbb{F}_q^k\setminus\set{0}\\ M\in\mathcal{M}}} \widehat{F}(0, M\alpha_0)\chi_{0, M\alpha_0}(s,x)
  +\sum\limits_{\substack{\alpha_0\in\mathbb{F}_q^k\setminus\set{0}\\M\in\mathbb{F}_q^{\ell}}} \widehat{F}(\alpha_0,0,\ldots,0)\chi_{\alpha_0,0,\ldots,0}(s,x).
  \end{equation}
  We plug in~\eqref{eq_1},~\eqref{eq_2},~\eqref{eq_3} into~\eqref{eq:equate_coeff} and equate coefficients to get the statement of the claim.
\end{proof}

\begin{corollary}\label{corr:up_fourier_coef}
  Suppose $S^{\star}$ is $\xi$-pseudo-random against zoom-out as well as with respect to the linear part. Then for all $\alpha$,
  \[
  \card{\widehat{f_{1,{\sf lin}}}(\alpha)}\leq \frac{1}{(q-1)(q^{\ell}-1)}\xi,
  \qquad\qquad
  \card{\widehat{f_{1,{\sf aff}}}(\alpha)}\leq \frac{\xi}{q^{\ell+1}}.
  \]
\end{corollary}
\begin{proof}
  We begin with the first inequality. From Claim~\ref{claim:f_Fourier} we have
  \[
  \widehat{f_{1,{\sf lin}}}(\alpha)
  =\frac{1}{q-1}\Expect{s,x}{F(s,x)\chi_{0,\alpha,\ldots,\alpha}(s,x)}
  =\frac{1}{q-1}\Expect{s,x,M\in\mathcal{M}}{F(s,Mx)\chi_{0,\alpha,\ldots,\alpha}(s,M x)}.
  \]
  Using the symmetries of $F$ we have that this is equal to
  \[
  \frac{1}{q-1}\Expect{s,x}{F(s,x)\Expect{M\in\mathcal{M}}{\chi_{0,\alpha,\ldots,\alpha}(s,M x)}}
  =\frac{1}{(q-1)(q^{\ell}-1)}\Expect{s,x}{F(s,x)\left(\sum\limits_{M\in\mathbb{F}_q^{\ell}}{\chi_{0,\alpha,\ldots,\alpha}(s,M x)}-1\right)},
  \]
  where in the last transition we turned expectation into sum and added/substracted $M=0$. Note that the sum over $M$ is $q^{\ell}$ if
  $\inner{x_i}{\alpha} = 0$ for all $i$ and $0$ otherwise, so the last expression is equal to
  \[
  \frac{1}{(q-1)(q^{\ell}-1)}\left(\Expect{s,x}{F(s,x)q^{\ell} 1_{{\sf span}(x)\subseteq W_{\alpha}}} - \mu(S^{\star})\right),
  \]
  where $W_{\alpha}$ is the subspace $\sett{z\in\mathbb{F}_q^{k}}{\inner{z}{\alpha} = 0}$. This is equal to
  \[
  \frac{1}{(q-1)(q^{\ell} - 1)}\left(\mu((S^{\star})_{W_{\alpha}, {\sf lin}}) - \mu(S^{\star})\right).
  \]
  The result now follows from the pseudo-randomness of $S^{\star}$ with respect to zoom-outs.

  \skipi
  We now move on to the second inequality. For $\alpha\in\mathbb{F}_q^k\setminus\set{0}$ and $j\in\mathbb{F}_q$, denote
  \[
  W_{\alpha,j} = \sett{z\in\mathbb{F}_q^{k}}{\inner{z}{\alpha} = j}.
  \]
  Then
  \begin{align*}
  1_{z\in W_{\alpha,j}}
  = \sum\limits_{v\in\mathbb{F}_q}\chi_{v}(\inner{z}{\alpha} - j)
  = \sum\limits_{v\in\mathbb{F}_q}\chi_{v}(-j)\chi_{v}(\inner{z}{\alpha})
  &= \sum\limits_{v\in\mathbb{F}_q}\chi_{v}(-j)\omega^{{\sf Tr}(v(z_1\alpha_1 + \ldots + z_k \alpha_k))}\\
  &= \sum\limits_{v\in\mathbb{F}_q}\chi_{v}(-j)\prod\limits_{i=1}^{k}\omega^{{\sf Tr}(v z_i\alpha_i)}\\
  &= \sum\limits_{v\in\mathbb{F}_q}\chi_{v}(-j)\chi_{v\alpha}(z).
  \end{align*}
  We now invert this formula. We multiply this equality by $\chi_{1}(j)$ and average over $j$ to get that
  \begin{equation}\label{eq_4}
  \chi_{\alpha}(z) = \frac{1}{q}\sum\limits_{j\in\mathbb{F}_q}\chi_{1}(j)1_{z\in W_{\alpha,j}},
  \end{equation}
  and we use this equality to establish the second inequality of the lemma.
  \begin{align*}
  \widehat{f_{1,{\sf aff}}}(\alpha)
  =\widehat{F}(\alpha,0,\ldots,0)
  =\Expect{s,x}{F(s,x)\chi_{\alpha}(s)}
  &=\Expect{s,x,M}{F(s-Mx, x)\chi_{\alpha}(s)}\\
  &=\Expect{s,x}{F(s,x)\Expect{M}{\chi_{\alpha}(s+Mx)}}\\
  &=\Expect{s,x}{F(s,x)\chi_{\alpha}(s)\Expect{M}{\chi_{\alpha}(Mx)}}.
  \end{align*}
  As before, the expectation over $M$ is $1$ if $\inner{\alpha}{x_i} = 0$ for all $i$ and $0$ otherwise, so
  \[
  \widehat{f_{1,{\sf aff}}}(\alpha)
  =\Expect{s,x}{F(s,x)\chi_{\alpha}(s)1_{{\sf span}(x)\subseteq W_{\alpha,0}}}.
  \]
  Plugging in~\eqref{eq_4} now yields
  \begin{align*}
  \widehat{f_{1,{\sf aff}}}(\alpha)
  =\Expect{j\in\mathbb{F}_q}{\chi_1(j) \Expect{s,x}{F(s,x)1_{s\in W_{\alpha,j}} 1_{{\sf span}(x)\subseteq W_{\alpha,0}}}}
  &=\Expect{j\in\mathbb{F}_q}{\chi_1(j) \Expect{s,x}{F(s,x)1_{s+{\sf span}(x)\subseteq W_{\alpha,j}}}}\\
  &=\Expect{j\in\mathbb{F}_q}{\chi_1(j) q^{-(\ell+1)}\mu((S^{\star})_{W_{\alpha,j}})}.
  \end{align*}
  Taking absolute value, applying the triangle inequality and using the pseudo-randomness of $S^{\star}$ finishes the proof.
  %The second inequality is very similar, and we briefly outline it. From Claim~\ref{claim:f_Fourier} we have
%  \[
%  \widehat{f_{1,{\sf aff}}}(\alpha)
%  =\Expect{s,x}{F(s,x)\chi_{\alpha,0,\ldots,0}(s,x)}
%  =\Expect{s,x,M\in\mathcal{M}}{F(s-Mx,x)\chi_{\alpha,0,\ldots,0}(s,x)},
%  \]
%  where we used the symmetries of $F$. Making the change of variables $s\leftarrow s-Mx$ we get that
%  \begin{align*}
%  \widehat{f_{1,{\sf aff}}}(\alpha)
%  =\Expect{s,x,M\in\mathcal{M}}{F(s,x)\chi_{\alpha,0,\ldots,0}(s+Mx,x)}
%  &=\Expect{s,x}{F(s,x)\Expect{M}{\chi_{\alpha,0,\ldots,0}(s+Mx,x)}}\\
%  &=\Expect{s,x}{F(s,x)\Expect{M}{\chi_{\alpha,\alpha,\ldots,\alpha}(s,Mx)}}
%  \end{align*}
\end{proof}
\section{Proof of Lemma~\ref{lem:4th_norm_ub}}
In this section we prove Lemma~\ref{lem:4th_norm_ub}. The proof proceeds by opening up the $4$-norm and upper bounding different terms in an appropriate way.
Write
\[
g(s,x) = \sum\limits_{M\in\mathcal{B}}{f_{1,{\sf lin}}(\inner{M}{x})},
\qquad\qquad
h(s,x) = \sum\limits_{M\in\mathbb{F}_q^{\ell}}{f_{1,{\sf aff}}(s+\inner{M}{x})}.
\]
Clearly
\begin{equation}\label{eq_6}
F_1(s,x)^4 =  g(s,x)^4 + 4g(s,x)^3h(s,x) + 6g(s,x)^2h(s,x)^2 + 4g(s,x) h(s,x)^3 + h(s,x)^4,
\end{equation}
and we prove that the expectation of all but the last term is very small. As we will see, it is enough for us to upper bound
the expectation of $g(s,x)^4$ and $h(s,x)^4$, but we remark that it is possible to directly analyze each one of these terms separately
in order to establish better bounds.

\begin{claim}\label{claim:up_cl2}
  $\Expect{s,x}{g(s,x)^4} \leq \xi^2\mu(S^{\star}) + 4(q-1)^2\xi\mu(S^{\star}) +  24\frac{\xi^2}{(q-1)^2}\mu(S^{\star})$. In particular,
  we have that $\Expect{s,x}{g(s,x)^4}\leq 30q^2\xi\mu(S^{\star})$.
\end{claim}
\begin{proof}
  We open up according to the definition of $g(s,x)$:
  \begin{align*}
  g(s,x)^4
  &=\sum\limits_{M_1,M_2,M_3,M_4\in\mathcal{B}}{f_{1,{\sf lin}}(\inner{M_1}{x})\cdots f_{1,{\sf lin}}(\inner{M_4}{x})}\\
  &=\frac{1}{(q-1)^4}\sum\limits_{M_1,M_2,M_3,M_4\in\mathcal{M}}{f_{1,{\sf lin}}(\inner{M_1}{x})\cdots f_{1,{\sf lin}}(\inner{M_4}{x})}.
  \end{align*}
  We partition the last sum according to ${\sf dim}({\sf span}(M_1,\ldots,M_4))$. Denote by $H_i$ the collection of $(M_1,\ldots,M_4)$ for which
  this dimension is $i$.
  \paragraph{The contribution from $H_1$.}
  Note that the summands corresponding to $H_1$ may be written as
  \begin{align*}
  &\frac{1}{(q-1)^4}\sum\limits_{M_1\in\mathcal{M},M_2,M_3,M_4\in {\sf span}(M_1)\setminus\set{0}}{f_{1,{\sf lin}}(\inner{M_1}{x})\cdots f_{1,{\sf lin}}(\inner{M_4}{x})}
  &=\frac{1}{q-1}\sum\limits_{M\in\mathcal{M}}f_{1,{\sf lin}}(\inner{M_1}{x})^4.
  \end{align*}
  Taking expectation over $x$ we get that the contribution from $H_1$ is at most
  \[
  \frac{1}{q-1}\card{M}\Expect{z}{f_{1,{\sf lin}}(\inner{M_1}{x})^4}
  \leq \card{\mathcal{B}}\norm{f_{1,{\sf lin}}}_{\infty}^2\norm{f_{1,{\sf lin}}}_{2}^2.
  \]
  Using Claim~\ref{claim:f_2nd} we bound $\norm{f_{1,{\sf lin}}}_{2}^2\leq \frac{\mu(S^{\star})}{\card{\mathcal{B}}}$,
  and using the $\xi$ pseudo-randomness of $S^{\star}$ with respect to zoom ins on the linear part we have $\norm{f_{1,{\sf lin}}}_{\infty}\leq \xi$,
  so the contribution from $H_1$ is at most $\xi^2\mu(S^{\star})$.

  \paragraph{The contribution from $H_2$.}
  There are two cases. Either we can partition $M_1,M_2,M_3,M_4$ into two sets, such that the dimension of the space spanned by each one is $2$, or we cannot.
  The contribution of the first type is at most
  \begin{align*}
  &\frac{1}{(q-1)^4}\sum\limits_{\substack{M_1,M_2\in\mathcal{M}\text{ linearly ind}\\ M_3,M_4\in {\sf span}(M_1,M_2)\text{ linearly ind}}}
  \card{f_{1,{\sf lin}}(\inner{M_1}{x})\cdots f_{1,{\sf lin}}(\inner{M_4}{x})}\\
  &\leq \frac{2}{(q-1)^4}\sum\limits_{\substack{M_1,M_2\in\mathcal{M}\text{ linearly ind}\\ M_3,M_4\in {\sf span}(M_1,M_2)\text{ linearly ind}}}
  \card{f_{1,{\sf lin}}(\inner{M_1}{x})f_{1,{\sf lin}}(\inner{M_2}{x})}^2 + \card{f_{1,{\sf lin}}(\inner{M_3}{x}) f_{1,{\sf lin}}(\inner{M_4}{x})}^2.
  \end{align*}
  Taking expectation, the contribution from $H_2$ is at most
  \[
  4\sum\limits_{M_1,M_2\in\mathcal{M}\text{ linearly independent}}\Expect{s,x}{\card{f_{1,{\sf lin}}(\inner{M_1}{x})f_{1,{\sf lin}}(\inner{M_2}{x})}^2}.
  \]
  As $\inner{M_1}{x}$ and $\inner{M_2}{x}$ are independently uniformly distributed in $\mathbb{F}_q^k$, we get that the last expression is
  \[
  4\card{\mathcal{M}}\norm{f_{1,{\sf lin}}}_2^4
  \leq 4\card{\mathcal{M}}\left(\frac{\mu(S^{\star})}{\card{\mathcal{B}}}\right)^2
  \leq 4(q-1)^2\xi\mu(S^{\star}),
  \]
  where we used Claim~\ref{claim:f_2nd}.

  The contribution of the second type is a multiple of
  \[
  \frac{1}{(q-1)^4}\sum\limits_{\substack{M_1,M_2\in\mathcal{M}\text{ linearly independent}}}
  f_{1,{\sf lin}}(\inner{M_1}{x})f_{1,{\sf lin}}(\inner{M_2}{x})f_{1,{\sf lin}}(\inner{M_1}{x})f_{1,{\sf lin}}(\inner{M_1}{x}),
  \]
  and taking expectation the contribution of this type is proportional to
  \[
  \frac{1}{(q-1)^4}\sum\limits_{\substack{M_1,M_2\in\mathcal{M}\text{ linearly ind}}}
  \Expect{s,x}{ f_{1,{\sf lin}}(\inner{M_1}{x})^3 f_{1,{\sf lin}}(\inner{M_2}{x})},
  \]
  which is equal to $0$ as $\inner{M_1}{x}$ and $\inner{M_2}{x}$ are uniform and independent in $\mathbb{F}_q^k$, and the
  expectation of $f_{1,{\sf lin}}(\inner{M_2}{x})$ is $0$ by Claim~\ref{claim:f_orth}.

  \paragraph{The contribution from $H_3$.}
  The contribution of this case is a constant multiple, not more than $4!$, of
  \[
  \frac{1}{(q-1)^4}\sum\limits_{\substack{M_1,M_2,M_3\in\mathcal{M}\text{ linearly ind}\\ M_4\in {\sf span}(M_1,M_2,M_3)}}
  \Expect{s,x}{f_{1,{\sf lin}}(\inner{M_1}{x})f_{1,{\sf lin}}(\inner{M_2}{x})f_{1,{\sf lin}}(\inner{M_3}{x})f_{1,{\sf lin}}(\inner{M_4}{x})}.
  \]
  If $M_4\in {\sf span}(M_1,M_2)$, the contribution is shown to be $0$ as in the second type in the analysis of $H_2$. Otherwise,
  we get
  \[
  \frac{1}{(q-1)^4}\sum\limits_{\substack{M_1,M_2,M_3\in\mathcal{M}\\\text{ linearly ind}\\ j_1,j_2,j_3\in\mathbb{F}_q\setminus\{0\}}}
  \Expect{s,x}{f_{1,{\sf lin}}(\inner{M_1}{x})f_{1,{\sf lin}}(\inner{M_2}{x})f_{1,{\sf lin}}(\inner{M_3}{x})f_{1,{\sf lin}}(\inner{j_1M_1+j_2M_2+j_3M_3}{x})}.
  \]
  Taking expectation, we get that the contribution is proportional to
  \[
  \frac{1}{(q-1)^4}\card{\set{M_1,M_2,M_3\in\mathcal{M}\text{ linearly ind}}}
  \hspace{-3ex}
  \sum\limits_{j_1,j_2,j_3\in\mathbb{F}_q\setminus\set{0}}
  \Expect{u,v,w}{f_{1,{\sf lin}}(u)f_{1,{\sf lin}}(v)f_{1,{\sf lin}}(w)f_{1,{\sf lin}}(j_1u+j_2v+j_3w)}.
  \]
  Taking the proportionality constant into consideration, and taking $j_1,j_2,j_3$ that maximize this expectation, the contribution from $H_3$ is at most
  \begin{equation}\label{eq_5}
  \frac{4!}{q-1} \card{\mathcal{M}}^3
  \card{\Expect{u,v,w}{f_{1,{\sf lin}}(u)f_{1,{\sf lin}}(v)f_{1,{\sf lin}}(w)f_{1,{\sf lin}}(j_1u+j_2v+j_3w)}},
  \end{equation}
  and to upper bound the last expectation we move to the Fourier domain. A straightforward computation shows that
  \begin{align*}
  &\card{\Expect{u,v,w}{f_{1,{\sf lin}}(u)f_{1,{\sf lin}}(v)f_{1,{\sf lin}}(w)f_{1,{\sf lin}}(j_1u+j_2v+j_3w)}}\\
  &=\card{\sum\limits_{\alpha}\widehat{f_{1,{\sf lin}}}(-j_1\alpha)\widehat{f_{1,{\sf lin}}}(-j_2\alpha)\widehat{f_{1,{\sf lin}}}(-j_3\alpha)\widehat{f_{1,{\sf lin}}}(\alpha)}.
  \end{align*}
  Using Claim~\ref{corr:up_fourier_coef} we get that this is at most
  \begin{align*}
   \frac{\xi^2}{(q-1)^2(q^{\ell} - 1)^2}\sum\limits_{\alpha}\card{\widehat{f_{1,{\sf lin}}}(-j_3\alpha)\widehat{f_{1,{\sf lin}}}(\alpha)}
   &\leq \frac{\xi^2}{(q-1)^2(q^{\ell} - 1)^2}\sum\limits_{\alpha}\card{\widehat{f_{1,{\sf lin}}}(\alpha)}^2\\
   &=\frac{\xi^2}{(q-1)^2(q^{\ell} - 1)^2}\norm{f_{1,{\sf lin}}}_2^2\\
   &\leq \frac{\xi^2 \mu(S^{\star})}{(q-1)(q^{\ell} - 1)^3},
  \end{align*}
  where in the last transition we used Claim~\ref{claim:f_2nd}. Plugging this into~\eqref{eq_5} yields that the contribution from $H_3$ is at most
  \[
  24\frac{\xi^2}{(q-1)^2}\mu(S^{\star}).
  \]
  \paragraph{The contribution from $H_4$.} This is shown to be $0$ similarly to the second type in the analysis of $H_2$.
\end{proof}

Next, we upper bound the expectation of $h(s,x)^4$.
\begin{claim}\label{claim:up_cl4}
  We have
  \[
  \Expect{s,x}{h(s,x)^4} \leq \mu(S^{\star})\norm{f_{1,{\sf aff}}}_{\infty}^2 + 32\xi\mu(S^{\star}) + \xi^2 q\mu(S^{\star}).
  \]
  In particular:
  \begin{enumerate}
    \item $\Expect{s,x}{h(s,x)^4}\leq \mu(S^{\star})\norm{f_{1,{\sf aff}}}_{\infty}^2 + 33q\xi\mu(S^{\star})$;
    \item and weakening further, $\Expect{s,x}{h(s,x)^4}\leq 34q\mu(S^{\star})$.
  \end{enumerate}
\end{claim}
\begin{proof}
  We open up according to the definition of $h(s,x)$:
  \[
  \Expect{s,x}{h(s,x)^4}
  =\sum\limits_{M_1,M_2,M_3,M_4\in\mathbb{F}_q^{\ell}}{
  \Expect{s,x}{f_{1,{\sf aff}}(s+\inner{M_1}{x})\cdots f_{1,{\sf aff}}(s+\inner{M_4}{x})}}.
  \]
  We make the change of variables $s\leftarrow s+\inner{M_1}{x}$ and get that
  \begin{align*}
  &\Expect{s,x}{h(s,x)^4}\\
  &=\hspace{-4ex}\sum\limits_{M_1,M_2,M_3,M_4\in\mathbb{F}_q^{\ell}}{
  \Expect{s,x}{f_{1,{\sf aff}}(s)f_{1,{\sf aff}}(s+\inner{M_2-M_1}{x})f_{1,{\sf aff}}(s+\inner{M_3-M_1}{x})f_{1,{\sf aff}}(s+\inner{M_4-M_1}{x})}}.
  \end{align*}
  We partition the last sum according to ${\sf dim}({\sf span}(M_2-M_1,M_3-M_1,M_4-M_1))$. For $i=0,\ldots,3$ denote by $H_i$ the collection of $(M_1,\ldots,M_4)$ for which
  this dimension is $i$.
  \paragraph{The contribution from $H_0$.}
  The contribution here is
  \[
  \sum\limits_{M_1\in\mathbb{F}_q^{\ell}}{
  \Expect{s,x}{f_{1,{\sf aff}}(s)^4}}
  \leq q^{\ell}\norm{f_{1,{\sf aff}}}_{\infty}^2\norm{f_{1,{\sf aff}}}_{2}^2.
  \]
  Using Claim~\ref{claim:f_2nd}, this is upper bounded by $\mu(S^{\star})\norm{f_{1,{\sf aff}}}_{\infty}^2$.

  \paragraph{The contribution from $H_1$}
  There are three subcases we consider. Either there are two differences, say $M_2-M_1$, $M_3-M_1$ which are $0$, in which case
  the contribution is
  \[
  \sum\limits_{M_1,M_4\in\mathbb{F}_q^{\ell}}{
  \Expect{s,x}{f_{1,{\sf aff}}(s)^3f_{1,{\sf aff}}(s+\inner{M_4-M_1}{x})}}.
  \]
  The points $s$ and $s+\inner{M_4-M_1}{x}$ are jointedly distributed uniformly on $\mathbb{F}_q^k$, so the expectation above may be broken into
  the product of two expectation, and the expectation of $f_{1,{\sf aff}}(s+\inner{M_4-M_1}{x})$ is $0$ by Claim~\ref{claim:f_orth}. Hence, the contribution
  of this sub-case is $0$.

  In the second subcase, $M_2 - M_1 = M_3-M_1 = M_4-M_1$, and the contribution here is $0$ just like in the previous subcase.
  In the last subcase, we consider $M_1,M_2,M_3$ that maximize the absolute value of the expectation and upper bound the contribution as
  \[
  q^{2\ell} q^2 \card{\Expect{s,x}{f_{1,{\sf aff}}(s)f_{1,{\sf aff}}(s+\inner{M_2-M_1}{x})f_{1,{\sf aff}}(s+\inner{M_3-M_1}{x})f_{1,{\sf aff}}(s+\inner{M_4-M_1}{x})}}.
  \]
  \begin{enumerate}
    \item  If one of the differences is $0$, say $M_2 - M_1 = 0$, then we conclude that $M_3 - M_1$ and $M_4-M_1$ are difference (otherwise we would have been in a previous
  subcase), and the contribution here is at most
  \begin{align*}
  &q^{2\ell} q^2 \card{\Expect{s,x}{f_{1,{\sf aff}}(s)^2 f_{1,{\sf aff}}(s+\inner{M_3-M_1}{x}) f_{1,{\sf aff}}(s+\inner{M_4-M_1}{x})}}\\
  &\leq 2 \cdot q^{2\ell} q^2
  \Big|\Expect{s,x}{f_{1,{\sf aff}}(s)^2 f_{1,{\sf aff}}(s+\inner{M_3-M_1}{x})^2}
  +\Expect{s,x}{f_{1,{\sf aff}}(s)^2f_{1,{\sf aff}}(s+\inner{M_4-M_1}{x})^2}\Big|.
  \end{align*}
  Each one of these expectations is equal to $\norm{f_{1,{\sf aff}}}_2^4$, so we get an upper bound of
  \[
  4 \cdot q^{2\ell} q^2 \norm{f_{1,{\sf aff}}}_2^4
  \leq 4\cdot q^{2\ell} q^2\left(\frac{\mu(S^{\star})}{q^{\ell}}\right)^2
  \leq 4q^2\xi \mu(S^{\star}),
  \]
  where we used Claim~\ref{claim:f_2nd}.
    \item Otherwise, all three differences are non $0$ and at least two are different, say $M_3-M_1\neq M_4-M_1$.
    We thus bound the contribution by
    \begin{align*}
    &q^{2\ell} q^2\Big|\E_{s,x}\Big[f_{1,{\sf aff}}(s)f_{1,{\sf aff}}(s+\inner{M_2-M_1}{x})\\
    &\qquad\qquad\cdot f_{1,{\sf aff}}(s+\inner{M_3-M_1}{x}) f_{1,{\sf aff}}(s+\inner{M_4-M_1}{x})\Big]\Big|\\
    &\leq 2q^{2\ell} q^2\Big|\E_{s,x}\Big[f_{1,{\sf aff}}(s)^2f_{1,{\sf aff}}(s+\inner{M_2-M_1}{x})^2\\
    &\qquad\qquad\qquad+ f_{1,{\sf aff}}(s+\inner{M_3-M_1}{x})^2 f_{1,{\sf aff}}(s+\inner{M_4-M_1}{x})^2\Big]\Big|.
  \end{align*}
  The last expectation is equal to $2\norm{f_{1,{\sf aff}}}_{2}^4$, so we get contribution of $4q^{2\ell} q^2\left(\frac{\mu(S^{\star})}{q^{\ell}}\right)^2
  \leq 4q^2\xi\mu(S^{\star})$.
  \end{enumerate}

  %We know that $M_2-M_1,M_3-M_1,M_4-M_1$ are not all the same and at most one of them is $0$
%
%  Taking expectation over $x$ we get that the contribution from $H_1$ is at most
%  \[
%  \frac{1}{q-1}\card{M}\Expect{z}{f_{1,{\sf lin}}(\inner{M_1}{x})^4}
%  \leq \card{\mathcal{B}}\norm{f_{1,{\sf lin}}}_{\infty}^2\norm{f_{1,{\sf lin}}}_{2}^2.
%  \]
%  Using Claim~\ref{claim:f_2nd} we bound $\norm{f_{1,{\sf lin}}}_{2}^2\leq \frac{\mu(S^{\star})}{\card{\mathcal{B}}}$,
%  and using the $\xi$ pseudo-randomness of $S^{\star}$ with respect to zoom ins on the linear part we have $\norm{f_{1,{\sf lin}}}_{\infty}\leq \xi$,
%  so the contribution from $H_1$ is at most $\xi^2\mu(S^{\star})$.

  \paragraph{The contribution from $H_2$.}
  Let $M_1,M_2,M_3,M_4$ that maximize this case. Then we need to bound
  \[
  q^{3\ell} q^3 \card{\Expect{s,x}{f_{1,{\sf aff}}(s)f_{1,{\sf aff}}(s+\inner{M_2-M_1}{x})f_{1,{\sf aff}}(s+\inner{M_3-M_1}{x})f_{1,{\sf aff}}(s+\inner{M_4-M_1}{x})}}.
  \]
  Suppose without loss of generality $M_2-M_1, M_3-M_1$ constitute a basis for ${\sf span}(M_2-M_1,M_3-M_1, M_4-M_1)$. Let $j_3,j_2$ be such that
  $M_4-M_1 = j_3(M_3-M_1) + j_2(M_2-M_1)$, and make the change of variables $u = s+\inner{M_2-M_1}{x}$, $w=s+\inner{M_3-M_1}{x}$ and note that
  $(s,u,v)$ are distributed uniformly on $(\mathbb{F}_q^{k})^3$. Thus, the above expectation is
  \[
  \Expect{s,u,w}{f_{1,{\sf aff}}(s)f_{1,{\sf aff}}(u)f_{1,{\sf aff}}(w)f_{1,{\sf aff}}((1-j_3-j_2)s + j_2u + j_3w)}.
  \]
  If $j_2 = 0$, $j_3 = 0$ or $j_2+j_3 = 1$, then this expectation is $0$. Indeed, say $j_2 = 0$, then $u$ only appears in the second term and is thus
  independent of the rest, and by Claim~\ref{claim:f_orth} its expectation is $0$. We thus assume otherwise, and move to the Fourier domain. A straightforward
  computation shows that
  \[
  \sum\limits_{\alpha}{\widehat{f_{1,{\sf aff}}}((j_2+j_3-1)\alpha)\widehat{f_{1,{\sf aff}}}(-j_2\alpha)\widehat{f_{1,{\sf aff}}}(-j_3\alpha)\widehat{f_{1,{\sf aff}}}(\alpha)}.
  \]
  Taking absolute value, the absolute value of this sum is at most
  \[
  \norm{\widehat{f_{1,{\sf aff}}}}_{\infty}^2
  \card{\sum\limits_{\alpha}{\widehat{f_{1,{\sf aff}}}(-j_3\alpha)\widehat{f_{1,{\sf aff}}}(\alpha)}}
  \leq\norm{\widehat{f_{1,{\sf aff}}}}_{\infty}^2
  \sum\limits_{\alpha}{\widehat{f_{1,{\sf aff}}}(\alpha)^2}
  \leq\norm{\widehat{f_{1,{\sf aff}}}}_{\infty}^2\norm{\widehat{f_{1,{\sf aff}}}}_{2}^2.
  \]
  Using Claim~\ref{claim:f_2nd} and Corollary~\ref{corr:up_fourier_coef} we may bound this by $\frac{\xi^2}{q^{3\ell+2}}\mu(S^{\star})$,
  and plugging this above we get that the contribution from $H_2$ is at most
  \[
  q^{3\ell} q^3\frac{\xi^2}{q^{3\ell+2}}\mu(S^{\star})
  =\xi^2 q\mu(S^{\star}).
  \]
  \paragraph{The contribution from $H_3$.}
  In this case, the joint distribution of $s$, $s+\inner{M_2-M_1}{x}$, $s+\inner{M_3-M_1}{x}$, $s+\inner{M_4-M_1}{x}$ is
  uniform over $(\mathbb{F}_q^k)^4$, so the contribution is $0$ by Claim~\ref{claim:f_orth}.
\end{proof}

\begin{claim}\label{claim:up_cl6}
  $\Expect{s,x}{4g(s,x)^3h(s,x) + 6g(s,x)^2h(s,x)^2 + 4g(s,x)h(s,x)^3}\leq 800q^2\xi^{1/4}\mu(S^{\star})$.
\end{claim}
\begin{proof}
Using Holder's inequality, we have
\[
\Expect{s,x}{4g(s,x)^3h(s,x)+6g(s,x)^2h(s,x)^2+4g(s,x)h(s,x)^3}\leq
4\norm{g}_4^3\norm{h}_4 + 6\norm{g}_4^2\norm{h}_4^2 + 4\norm{g}_4\norm{h}_4^3.
\]
Use Claim~\ref{claim:up_cl2} and the second item of Claim~\ref{claim:up_cl4} to bound each term on the right hand side,
we get that it is at most
\[
4(30q^2\xi\mu(S^{\star}))^{3/4}(34 q\mu(S^{\star}))^{1/4}
+6(30q^2\xi\mu(S^{\star}))^{1/2}(34 q\mu(S^{\star}))^{1/2}
+4(30q^2\xi\mu(S^{\star}))^{1/4}(34 q\mu(S^{\star}))^{3/4}.
\]
Further upper bounding this we get it is at most
\[
14 \cdot 34q^2\xi^{1/4}\mu(S^{\star})
\leq 800q^2\xi^{1/4}\mu(S^{\star}).\qedhere
\]
\end{proof}

We are now ready to prove Lemma~\ref{lem:4th_norm_ub}.
\begin{proof}[Proof of Lemma~\ref{lem:4th_norm_ub}]
  Take expectation over~\eqref{eq_6} and use Claims~\ref{claim:up_cl2},~\ref{claim:up_cl4} (first item) and~\ref{claim:up_cl6} to get that
  \[
  \norm{F_1}_4^4
  \leq 30q^2\xi\mu(S^{\star}) + 800q^2\xi^{1/4}\mu(S^{\star}) + \mu(S^{\star})\norm{f_{1,{\sf aff}}}_{\infty}^2 + 33q\xi\mu(S^{\star}),
  \]
  which impllies
  \[
   \norm{F_1}_4^4
  \leq 863q^2\xi^{1/4}\mu(S^{\star}) + \mu(S^{\star})\norm{f_{1,{\sf aff}}}_{\infty}^2.
  \]
  Finally, note that $\norm{f_{1,{\sf aff}}}_{\infty}\leq a$, so we conclude that
  \[
  \norm{F_1}_4^4\leq 863q^2\xi^{1/4}\mu(S^{\star}) + a^2\mu(S^{\star}).\qedhere
  \]
\end{proof}

\end{document}